\documentclass[11pt]{article}

\usepackage{natbib}
\usepackage{fullpage}
\usepackage{enumitem}

\usepackage{amsfonts,graphpap,amscd,mathrsfs,graphicx,lscape}
\usepackage{amssymb,amstext,xspace,amsmath}
\usepackage{float}

\usepackage[ruled,vlined]{algorithm2e}
\renewcommand{\algorithmcfname}{ALGORITHM}
\SetAlFnt{\small}
\SetAlCapFnt{\small}
\SetAlCapNameFnt{\small}
\SetAlCapHSkip{0pt}
\IncMargin{-\parindent}

\usepackage{color}              
\usepackage{mathtools}
\usepackage[title,header]{appendix}

\usepackage{color-edits}
\usepackage{bbm}
\addauthor{shipra}{green}
\addauthor{costis}{red}
\addauthor{vahab}{blue}

\makeatletter
\DeclareRobustCommand*\cal{\@fontswitch\relax\mathcal}
\makeatother

\usepackage{url}
\newcommand{\mye}{{mye}}

\newcommand{\prob}{\rho}
\newcommand{\mpr}{p^*}
\newcommand{\R}{\mathbb{R}}

\newcommand{\Ex}{\text{E}}

\newcommand{\Val}{V}

\DeclarePairedDelimiterX{\ExpArg}[1]{[}{]}{#1}

\newcommand{\trans}{Q}
\newcommand{\alloc}{x}

\newcommand{\price}{p}
\newcommand{\bid}{b}
\newcommand{\avgbid}{\overline{b}}
\newcommand{\val}{v}
\newcommand{\util}{u}
\newcommand{\newstate}{s}
\newcommand{\newstatespace}{\mathcal{S}}

\newcommand{\mech}{M}

\newcommand{\mthree}{\mech(\epsilon,\prob,\price)}
\newcommand{\const}{{R}}

\newcommand{\mthreemye}{\mech(\epsilon,\prob,\mpr, \const)}

\newcommand{\rev}{\mathrm{Rev}}
\newcommand{\revmye}{\rev^{\mathrm{Mye}}}

\newcommand{\revmthreemyek}[1][k]{\rev_{#1}^{\mthreemye}}

\newcommand{\simp}{{\mathbf{\Delta}}}

\newcommand{\sacomment}[1]{}
\newcommand{\delete}[1]{}
\newcommand{\indi}{\mathbbm{1}}

\newcommand{\isborder}{=\pm}

\usepackage[dvipsnames,usenames]{xcolor}
\usepackage{prettyref}
\usepackage[hypertexnames=false,colorlinks=true,pdfpagemode=Usenone,urlcolor=Blue,linkcolor=RoyalBlue,citecolor=OliveGreen,pdfstartview=FitH]{hyperref}

\newtheorem{theorem}{Theorem}[section]
\newtheorem{definition}{Definition}
\newtheorem{claim}{Claim}
\newtheorem{lemma}[theorem]{Lemma}

\newtheorem{proposition}[theorem]{Proposition}

\newcommand{\qed}{\mbox{\ \ \ }\rule{6pt}{7pt} \bigskip}

\newcommand{\comment}[1]{}
\newenvironment{proof}{\noindent{\em Proof:}}{\hfill\qed}
\newenvironment{proofsketch}{\noindent{\em Proof Sketch:}}{\qed}

\newenvironment{prevproof}[2]{\noindent {\sc {Proof of 
{#1}~\ref{#2}:}}}{$\hfill\qed$\vskip \belowdisplayskip}

\newcommand{\shouldWeRemove}[1]{}
\newcommand{\new}[1]{#1} 

\newcommand{\removeEC}[1]{}

\makeatletter
\newenvironment{oneshot}[1]{\@begintheorem{#1}{\unskip}}{\@endtheorem}
\makeatother

\begin{document}
\title{Robust Repeated Auctions under Heterogeneous Buyer Behavior}
\author{
Shipra Agrawal\thanks{IEOR Department, Columbia University. 
\tt{sa3305@columbia.edu}.}
\and Constantinos Daskalakis\thanks{EECS, Massachusetts Institute of 
Technology. \tt{costis@csail.mit.edu}.}
\and Vahab Mirrokni\thanks{Google Research, New York . 
\tt{mirrokni@google.com}.}
\and Balasubramanian Sivan \thanks{Google Research, New 
York. \tt{balusivan@google.com}.}
}

\date{}
\maketitle{}
\thispagestyle{empty}

\newcommand{\expostLB}{1.6}
\begin{abstract}
	We study revenue optimization in a repeated auction between a single seller 
	and	a single buyer. Traditionally, the design of  repeated auctions requires strong
	modeling assumptions about the bidder behavior, such as it being myopic, infinite lookahead, or some specific form of learning behavior. 	Is it possible to design mechanisms which are simultaneously optimal against a multitude of possible buyer behaviors? We answer this question by designing a simple state-based mechanism that is 
	simultaneously approximately optimal against a $k$-lookahead buyer for all $k$, a 
	buyer who is a no-regret learner, and a buyer who is a policy-regret 
	learner. Against each type of buyer our mechanism attains a constant 
	fraction of the optimal revenue attainable against that type of buyer. We complement our positive 
	results with almost tight impossibility 
	results, showing that the revenue approximation
	tradeoffs achieved by our mechanism for different lookahead attitudes are near-optimal.

\end{abstract}

\newpage
\setcounter{page}{1}

\section{Introduction}
\label{sec:intro}
Developing a theory of repeated
auctions that outlines the boundaries for what is and is not possible
is of both scientific and commercial significance. 
On the application side, it is partly motivated by online sale of display ads in ad exchanges via repeated auctions.
An essential difference that 
sets apart the repeated/dynamic setting from its
one-shot counterpart is the significantly higher revenue that is achievable in
the former. The key reason for this difference is simply that bundling may increase
revenue, and repeated interactions provide ample opportunities to bundle across
time. 

In its gross form, a dynamic mechanism that bundles across time could
simply demand the buyer to pay her entire surplus for $T$ future rounds, save some small $\epsilon$,  upfront for the promise of getting the item for free in all future rounds. 
A risk-neutral buyer  would have no choice but to accept this offer to get expected
 utility of $\epsilon$ 
or else get $0$ utility. Such mechanisms that threaten buyers to get either tiny or
$0$ utility have several drawbacks, the most prominent being that they force the
buyer to make a huge payment upfront which is unappealing. This motivated a 
string of recent work \cite{ADH16,MPTZ16a,BMP16,MPTZ17} proposing mechanisms 
that satisfy {\em per round ex-post individual 
rationality}, i.e. that the buyer's utility is non-negative in every round 
under his optimal strategy, rather than {\em interim individual rationality (IIR),} 
which only requires that the buyer's long-term expected utility is 
non-negative. 

Nevertheless, the ability of these dynamic mechanisms to extract high revenue 
depends crucially on several non-trivial assumptions:
\begin{enumerate}
\item The buyer completely understands the seller's mechanism. In particular, 
he understands, and can optimally respond to the consequences of his actions 
today on his utility $k$ rounds later, for all $k$; \label{ass0}
\item The buyer believes that the interaction with the seller will last for all future rounds; \label{ass1}
\item The buyer believes that the seller will stick to her proposed mechanism for all future rounds. \label{ass2}
\end{enumerate} 
In particular, the notion of `infinite look-ahead buyers' which is baked into the widely used concept of dynamic incentive-compatibility, requires that the buyer's action in every round takes into account the consequences of his action on his 
utility in \emph{all future rounds}, thus relying on all of the above assumptions. There are important practical reasons invalidating these assumptions. 
Firstly, the buyer may not be fully informed about and/or trust all the details of the seller's mechanism. Furthermore, cognitive/computational limitations or uncertainty about the future may prevent buyers from being infinite lookahead. 
In the context of online 
advertising, for example, given the number 
and variety of display ad exchanges in
the market, with credibility levels all across the spectrum, the buyers often 
don't trust that the seller will faithfully implement the announced mechanism 
\cite{ExchangeTransparency}. 

As a result, the seller often faces a buyer population that employs a variety of strategies, beyond perfectly rational infinite lookahead utility maximization, in order to maximize their {\it perceived} utility. Such a buyer could 
\begin{enumerate}
\item be myopic or more generally, have a limited lookahead, i.e., evaluate his decisions today only 
	based on their effect on the utility of $k$ future rounds (a $k$-lookahead buyer).
\item be a learner, i.e., completely disregard the seller's description of 
	the mechanism, and instead make his decisions through his favorite  
	learning algorithm using only his observed feedback so far.
\end{enumerate}
In face of such heterogeneous behaviors, the revenue-optimal solution for the seller is to 
have a tailored mechanism for each buyer behavior. 
However, there are strong reasons precluding the implementation of different mechanisms, each targeting a specific buyer behavior. Such discriminative targeting may be legally infeasible, and it may also be practically infeasible, as it could 
be hard for the seller to identify a buyer's response behavior. The latter may 
not even be well-defined, as buyers may change their response strategy across 
time. These observations motivates us to we ask the following question. 
\begin{center}
{\em Can we design mechanisms which are robust against heterogeneous buyer behaviors? }
\end{center}
Specifically, we seek a {\it single} mechanism that gets approximately optimal revenue 
simultaneously against buyers with different lookahead and learning behaviors, 
i.e., against each type of buyer, obtains a constant  fraction of the optimal 
revenue achievable by mechanisms tailored to that specific type of buyer. 

{\paragraph{Our setting and main results.}
We study a repeated interaction between a single seller and a single buyer over multiple rounds. At the beginning of 
each round $t = 1,2,\ldots, T$, there is a 
single fresh good for sale whose private value $\val_t \in \Val$ 
for the buyer is drawn from a publicly known distribution\footnote{As we 
explain later, for our positive results, the seller only needs to know the mean 
$\mu$ and not the whole 
distribution $F$.} $F$ with finite expectation $\mu$. 
The buyer observes the valuation $\val_t$ and makes a bid $b_t$. The 
good for sale in round $t$ has to be either allocated to the buyer or discarded 
immediately (i.e., not 
carried forward). The buyer's valuations are additive across rounds.
}

{
Our goal is to investigate the sensitivity of revenue extraction to variations 
in both the `lookahead attitude' and the `learning behavior' of the buyer. To this end, we introduce a novel framework for characterizing  revenue tradeoffs of dynamic mechanisms under heterogeneous buyer behavior. We formalize variations in forward-planning attitude (or lack thereof) of the buyer by considering a range of lookahead levels: from myopic to limited $k\ge 1$ lookahead to infinite lookahead. We model learning behavior of the buyer using the popular concept of no-regret learning, and allow different levels of learning sophistication by considering buyers who minimize simple regret vs. policy regret. Formal definitions of these different buyer behavior are provided in Section \ref{sec:prelim}. We characterize a mechanism as robust if it simultaneously achieves near-optimal revenue for different buyer behaviors. 
\begin{definition}
We call a mechanism {\bf $(\alpha, \beta)$-robust against lookahead attitudes} if, for any per-round valuation distribution $F$, it simultaneously achieves an expected average (over $T$ rounds) revenue of at least
\begin{itemize}
\item $\alpha \cdot \revmye - o(1)$ against  any myopic buyer, and  
\item $\beta \cdot \mu - o(1)$ against any $k$-lookahead buyer with $k\ge 1$, 
\end{itemize}
where $\revmye = \max_x x(1-F(x))$, $\mu = \Ex_{x \sim F}[x]$, and $o(1)$ are functions that go to $0$ with $T$.
\end{definition}
\begin{definition}
We call a mechanism {\bf $(\alpha, \beta)$-robust against learning behaviors} if, for any per-round valuation distribution $F$, it simultaneously achieves an expected average (over $T$ rounds) revenue of at least
\begin{itemize}
\item $\alpha \cdot \revmye - o(1)$ against  any no-simple-regret learner, and  
\item $\beta \cdot \mu - o(1)$ against any no-policy-regret learner, 
\end{itemize}
where Myerson revenue $\revmye$ and mean $\mu$ are defined as in the previous definition, and $o(1)$ are functions that go to $0$ with $T$. 
\end{definition}
 
Our main contribution is a simple mechanism that can be tuned robust against different lookahead attitudes, as well as different learning behaviors. Our base mechanism will additionally satisfy interim individual rationality (IIR) and non-payment forcefulness, as defined in Section~\ref{sec:mechProp}. We will subsequently modify our base mechanism to improve its individual rationality properties.
\begin{theorem}
\label{th:1} 
For any $\epsilon\in (0,1)$, there exists a non-payment forceful and IIR mechanism that is $(\frac{\epsilon}{2}, 1-\epsilon)$-robust against {\it lookahead attitudes}.
\end{theorem}
\begin{theorem} \label{th:2}
For any $\epsilon\in (0,1)$, there exists a non-payment forceful and IIR mechanism that is $(\frac{\epsilon}{1+\epsilon}, 1-\epsilon)$-robust against {\it learning behaviors}.
\end{theorem}
Note that our robust mechanism navigates a  tradeoff between revenue achieved for different types of buyers, i.e., a tradeoff between $\alpha$ and $\beta$. Upfront, it is not obvious if $\alpha$ must be decreased to increase $\beta$. We prove an impossibility result showing that any dynamic mechanism must face such a tradeoff. Further, our impossibility result provides a quantitative limitation on the revenue tradeoff achievable,  demonstrating that our mechanism achieves the best tradeoff possible within a constant factor. 
\begin{theorem} \label{thm:lowerbound}
There exists a regular and decreasing hazard rate distribution $F$ such that, for 
all $\epsilon \in [0,1]$ and $\delta>1-\frac{\epsilon}{2}$, there is no mechanism that is non-payment forceful 
and $(\epsilon, \delta)$-robust against lookahead attitudes.
\end{theorem}
}
The proofs of Theorems~\ref{th:1}--\ref{th:2} are provided in Section~\ref{sec:revenue tradeoffs} and Appendix~\ref{app:lookahead}, while the proof of Theorem~\ref{thm:lowerbound} is provided solely in Appendix~\ref{app:lookahead}.

\paragraph{Injecting Ex-Post Individual Rationality}
The base mechanism provided thus far is robust with respect to different buyer behaviors, but only satisfies Interim Individual Rationality. Is it possible to strengthen this mechanism to also satisfy the stringent notion of per-round ex-post Individual Rationality, discussed earlier in the introduction and defined formally in Section~\ref{sec:mechProp}?
As it turns out, per-round ex-post IR is too strong a requirement when it comes to extracting revenue that is close to full surplus from forward-looking buyers with limited lookahead. Specifically, we show that there  exist valuation distributions for which the revenue achieved by any per round ex-post IR mechanism is exponentially smaller compared to full surplus for finite lookahead buyers. 
\begin{theorem}
\label{th:expostLB}
There exists a valuation distribution $F$ such that any non-payment forceful and per-round ex-post IR mechanism can achieve expected average (over rounds) revenue of at most  
$\log(k\mu)+1$ from a $k$-lookahead buyer with any finite $k\ge 1$.
\end{theorem}

In view of this limitation for finite lookahead buyers, we turn our attention to myopic buyers and infinite lookahead buyers, providing a modification to our base mechanism that guarantees per-round ex-post IR for these types of buyers. To be precise, for infinite lookahead buyers our mechanism guarantees per-round ex-post IR with high probability, i.e., there exists a strategy for the buyer that with high probability is both optimal and satisfies the per-round ex-post IR conditions. We provide the modification to our mechanism and establish its properties in Appendix~\ref{sec:inf-lookahead}, establishing the following:
\begin{theorem}
\label{thm:expost-ir-inf}
For any $\epsilon\in (0,1)$, there exists a non-payment forceful and IIR 
mechanism that is $(\frac{\epsilon}{2}, 1-\epsilon)$-robust against {\it 
lookahead attitudes}, ex-post IR against a myopic buyer, and with a high 
probability ex-post IR against an 
infinite lookahead buyer. 
\end{theorem}
The per-round ex-post IR property, even if only guaranteed with high 
probability for infinite lookahead buyers, is a strong requirement that 
successfully eliminates  advanced selling or threat-based mechanisms. For 
example, let us consider the advanced selling mechanism that in every round 
asks the buyer to pre-pay for tomorrow's good a price equal to its expected 
value, for the guarantee that this good will be allocated tomorrow. This 
mechanism gets expected average per round revenue equal to the mean of the 
distribution, but it violates high probability ex-post IR, since a 
forward-looking buyer will choose 
to pre-pay for tomorrow's good in every round. This results in realized utility 
of $v-\mu$, which may be negative  with constant probability.

\section{Our framework}
\label{sec:prelim}

In this section, we present the main components of our framework, including a {\it state-based mechanism design setting}, {\it formal models for buyers' heterogeneous lookahead and learning behaviors}, and a {\it revenue optimization objective} to calibrate the mechanisms against different behaviors.
\subsection{Setup}
\label{sec:prelim-setup}
We study a repeated interaction between a single seller and 
a single buyer for a finite number of rounds $T$. At the beginning of 
each round 
$t = 1,2,\ldots,T$, there is a 
single fresh good for sale whose private value $\val_t \in \Val$ 
 for the buyer 
is drawn from a publicly known distribution $F$ with 
expectation $\mu$. 
The buyer observes his value $\val_t$ at the beginning of the round and makes a bid $b_t$. The 
good for sale in round $t$ has to be either allocated to the buyer or discarded 
immediately (i.e., not carried forward). The buyer's valuations are additive across rounds. 

\paragraph{Round outcome.} The outcome of the game in round $t$ is a pair 
$(x_t,p_t)$, where $x_t \in 
\{0,1\}$ indicates whether or not the buyer received the good and $\price_t \in 
\R$ 
is the payment made by the buyer to the seller. For the outcome pair $(x_t,\price_t)$ in round $t$, the linear 
utility $\util_t$ of the buyer in round $t$  is given 
by $\util_t(b_t) = \val_t\alloc_t - \price_t$, \new{and the seller's revenue is given by $p_t$.}

\paragraph{Dynamic state-based mechanism.} We consider dynamic mechanisms, where the allocation and payment $(x_t, p_t)$, in a round $t$, may depend on the current state $s_t$, in addition to the bid $b_t$ made by the buyer. Such a dynamic mechanism $\mech$ is defined as a 5-tuple  
$\mech = (\newstatespace, \trans,\alloc,\price, s_1)$, where 
\begin{enumerate}
\item $\newstatespace$ is the state space over which the mechanism operates. 
The state space $\newstatespace$ can be finite dimensional or countably infinite 
dimensional. The cardinality of $\newstatespace$ can be finite, countably infinite 
or uncountably infinite. 

\item $\alloc: \newstatespace \times \R^+ \rightarrow \simp^{\{0,1\}}$ is 
the (randomized) allocation 
function, which at the beginning of round $t$ receives as input the state $\newstate_t 
\in \newstatespace$ in round $t$, the bid $\bid_t \in \R^+$ made by the bidder
in round $t$, and outputs the allocation $\alloc_t \in \{0,1\}$ for the bidder in round $t$, where 
$ \alloc_t \sim \alloc(\newstate_t,\bid_t).$

\item $\price: \newstatespace \times \R^+ \times \{0,1\} \rightarrow \R$ is 
the payment function, which at the beginning of round $t$ takes as input the state $\newstate_t
\in \newstatespace$ in round $t$, the bid $\bid_t \in \R^+$ made by the bidder in 
round $t$, and the allocation $\alloc_t$ sampled as above,
and outputs the payment $\price_t$ for the bidder in round $t$, 
i.e., $ \price_t = \price(\newstate_t,\bid_t,\alloc_t).$
\removeEC{
We restrict our mechanisms to require that the payment $\price_t$ is always non-negative, and $0$ when the bid is $0$, i.e.,
\begin{equation}
\label{eq:non-payment-forceful}
\price(\newstate_t,\bid_t,\alloc_t) \ge 0, \text{ and } \price(\newstate_t, 0,\alloc_t) = 0, \text{ for any } \newstate_t, \bid_t, \alloc_t.
\end{equation}
 This ensures that the mechanisms are {\em non-payment forceful}, as they cannot be forced to pay the buyer in any state and they cannot force payments out of buyers who bid $0$. 
}
\item $Q:\newstatespace\times\R^+ \times\{0,1\} \times \R
\rightarrow \simp^\newstatespace$ is a state-transition function that takes as input at 
the end of round $t$ the state $\newstate_t \in \newstatespace$, the bid $\bid_t \in 
\R^+$ of the bidder in round $t$, the allocation 
$\alloc_t \in \{0,1\}$ and the 
payment $\price_t \in \R$, and 
outputs the distribution of next state $\newstate_{t+1}$ for round $t+1$, 
i.e., $\newstate_{t+1} \sim Q(\newstate_t,\bid_t,\alloc_t,\price_t)$. \footnote{Given that our price function is a deterministic function of $\newstate_t, \bid_t, \alloc_t$ we could have also suppressed $\price_t$ from the arguments of $Q$.} 
\item $s_1 \in \newstatespace$ is the starting state.
\end{enumerate} 
\paragraph{Bidding strategy.}
We consider multiple types of buyer behaviors, namely, different lookahead attitudes and learning behaviors. Therefore, the buyer's `optimal' bidding strategy is not necessarily the one that maximizes linear utility $u_t(b) = v_t x_t - p_t$ in round $t$. Instead the choice of bids $\{b_t\}$ is determined by the behavioral setting, and we refer to them as `behaviorally optimal' choice of bids.

\subsection{Lookahead behaviors}
\paragraph{$\ell$-lookahead utility.} We will define lookahead attitudes of buyers using the concept of $\ell$-lookahead utility, i.e., the total utility over the current and next $\ell$ rounds. Assuming there are at least $\ell$ remaining rounds following round $t$, the buyer evaluates a bid $b$ in round $t$ by computing its expected utility over the current round plus the maximum expected utility obtainable over the next $\ell$ rounds. More precisely, 
at round $t$, given the current state $s_t$ and valuation $v_t$, for any $\ell$, the buyer's expected $\ell$-lookahead utility for bid $b$, assuming $T\ge t+\ell$, is defined as: 
\begin{eqnarray}
U_{\ell}^{t}(\newstate_t,\val_t,b) & := & \Ex[u_t(b) + \sup_{b_{t+1}, \ldots, b_{t+\ell}} \sum_{j=1}^\ell u_{t+j}(b_{t+j})] \\
& = & \Ex \left[\val_t \ \alloc_t(b) - \price_t(b) + \sup_{b'} U_{\ell-1}^{t+1}(\newstate_{t+1},\val_{t+1}, b') \bigg\lvert \newstate_t,\val_t\right], \label{eq:Uell}\\
U_0^t(\newstate_t,\val_t,b) & := & \Ex \left[\val_t \ \alloc_t(b) - \price_t(b)  \bigg\lvert \newstate_t,\val_t\right], \label{eq:U0}
\end{eqnarray}
 where 
$\alloc_t(b) \sim \alloc(\newstate_t,b), \price_t(b) = p(\newstate_t, b, x_t(b))
 $, $s_{t+1} \sim Q(s_t, b, x_t(b), p_t(b))$. And, the expectation was taken over the random values $\val_{t+1}\sim F$ and any randomization in the mechanism. 

\paragraph{Buyer lookahead behaviors.} We define the following types of buyers with different lookahead attitudes:
\begin{itemize}
\item {$k$-lookahead buyer:} A {\em $k$-lookahead buyer} is a buyer who, in every round $t$, picks his bid $b_t$ to maximize $\min\{k, T-t\}$-lookahead utility, i.e., a behaviorally-optimal bid for such a buyer in round $t$ is given by
\begin{equation}
\label{eq:kLookaheadBuyer}
b_t \in \arg \max_{b\in \R^+} U_{\min\{k,T-t\}}^{t}(\newstate_t,\val_t,b)
\end{equation}
We refer to the bid $b_t$ computed above as a {\bf $k$-lookahead optimal bid} (Note that in general, the maximizer in the above equation may not exist, see the technical remark below). 
Two special cases of $k$-lookahead buyers are:
\begin{itemize}
\item {Myopic buyer:} We refer to a $0$-lookahead buyer as a {\it myopic buyer}.
\item {Infinite-lookahead buyer:} If a mechanism lasts for $T$ rounds, we refer to a $k$-lookahead buyer with $k = T-1$ as an {\em infinite-lookahead buyer}. 
\end{itemize}
\item {Forward-looking buyer:} A {\em forward looking buyer} maximizes $k_t$-lookahead utility for some $k_t\in \{1, \ldots, T-t\}$, at every time $t$. Thus, a forward-looking buyer may use different lookaheads at different time steps, but always looks ahead at least one step. 
\end{itemize}
We say that a buyer is using a {\bf behaviorally-optimal policy} if the buyer's bid $b_t$ satisfies \eqref{eq:kLookaheadBuyer} for $k=k_t$ at all time steps $t$. 

\paragraph{Technical Remark.} To be completely formal, the definition of 
$k$-lookahead buyer given above requires that the maximizer in 
\eqref{eq:kLookaheadBuyer} exists. There are mechanisms in which this maximizer 
does not exist. For instance, consider a single-state mechanism that offers the 
item at a price of $\epsilon$, whenever the bid is $1-\epsilon<1$, and does not 
offer the item when the bid is $1$. This mechanism has no optimal $k$-lookahead 
bid for any $k$. Such mechanisms are undesirable as they make it difficult for 
the buyers to decide what bid to use, and therefore for the sellers to 
understand what revenue to expect, even if they know their buyer's lookahead 
attitude. For this reason, the mechanisms that we construct are such that there 
always exists an optimal  $k$-lookahead bid. 

\subsection{Learning behaviors}

\paragraph{Background on no-regret learning.} To formally model a broad class of buyer 
learning behaviors, we use the concept of {\it no-regret learning}, a widely 
studied solution concept in the context of  $T$ round online prediction problem 
with advice from $N$ experts. In this problem, at every round $t=1,\ldots, T$, 
an adversary picks reward ${\bf g_t}=\{g_{1,t}, \ldots, g_{N,t}\}$ where 
$g_{i,t}$ is the reward associated with expert $i$. The learner needs to pick 
an expert $i_t\in [N]$ to obtain reward 
$g_{i_t,t}$. Regret in time $T$ is defined as
\begin{equation}
\label{eq:Regret}
\text{Regret}(T) = \max_{i\in[N]} \sum_{t=1}^T g_{i,t} - \sum_{t=1}^T g_{i_t,t}.
\end{equation}
A {\it no-regret online learning algorithm} for this problem uses the past 
observations $g_{i_1, 1}, \ldots, g_{i_{t-1}, t-1}$ to make the decision $i_t$ 
in every round $t$ such that $Regret(T) \le o(T)$.  When the number of experts 
$N$ is finite, there are efficient and natural algorithms (e.g., EXP3 algorithm 
based on multiplicative weight updates) that achieve $O(\sqrt{NT\log N})$ 
regret. 

Note that 'regret' compares the total reward achieved by the learner to the reward of best {\it single} expert in hindsight.
Furthermore, even if the adversary is adaptive (i.e., generates ${\bf g_t}$ adaptively based on $i_1, \ldots, i_t$), in this simple-regret framework, the performance of the best expert is evaluated over the sequence of inputs ${\bf g_1, \ldots, g_T}$ produced by the adversary in response to the learner's decision, and not those that {\it would be produced} if this expert was used in all rounds.
This is an important distinction between the above definition of regret, for which efficient online learning algorithms like EXP3 are known, vs. the more sophisticated `{\it policy regret}' which we define next. 

\paragraph{Background on policy regret learning.}  A no-policy-regret learning algorithm is a 
more sophisticated learner based on the definition of policy regret from 
\cite{ADT12}. Such an algorithm faces an adaptive adversary, 
and achieves $o(T)$ {\it policy regret}, defined as:
\begin{equation}
\label{eq:policyRegret}
\text{Policy-regret}(T)= \max_{j_1,\ldots, j_T \in {\cal C}_T} \sum_{t=1}^T g_{j_t, t}(j_1, \ldots, j_{t}) - \sum_{t=1}^T g_{i_t, t}(i_1, \ldots, i_{t}).
\end{equation}
where ${\cal C}(T)$ is some benchmark class of deterministic sequences of 
experts of length $T$, and the $g_{j_t,t}$ and $g_{i_t,t}$ have 
been explicitly written as a function of past decisions to indicate adaptive 
adversarial response to the {\it sequence} of choices so far. A special case is 
where ${\cal C}(T)$ is the class of {\it single} expert sequences, so that 
\begin{equation}
\label{eq:policyRegret1}
\text{Policy-regret}(T)= \max_{i} \sum_{t=1}^T g_{i, t}(i, \ldots, i) - \sum_{t=1}^T g_{i_t, t}(i_1, \ldots, i_{t}).
\end{equation}

\paragraph{Our characterization of buyer learning behaviors.} We consider a buyer who only gets to observe  whether the current state is good ($s_t\in \bot$) or bad ($s_t \notin \bot)$, and the valuation $v_t$, before making the bid, and the outcome (allocation, price) after making the bid, but does not know (or does not trust) anything else about the seller's mechanism. 
Using these observations, the buyer is trying to decide bids $b_t$, using a learning algorithm under the experts learning framework described above. 
We formalize the notion of different levels of learning sophistication among buyers by considering two classes of learners: 
\begin{itemize}
\item {\it No-regret learner:} Such a buyer considers, as experts, a finite collection ${\cal E}$ of mappings from the state information ($s_t = \bot$ or $s_t \ne \bot$) and valuation $\val_t$ to a bid, i.e., set of experts 
\begin{equation}
\label{eq:experts}
{\cal E} = \{f:[\bot, \not\perp] \times \bar{V} \rightarrow \bar{V}\};
\end{equation}
where $\bar{V}$ is an  discretized (to arbitrary accuracy) version of $V$, in order to obtain a finite set of experts. 
 The buyer uses a no-regret learning algorithm to decide which expert $f_t\in {\cal E}$ to use in round $t$ to set $b_t=f_t(s_t, v_t)$.  The adversarial reward at time $t$ is given by the buyer's $t^{th}$-round utility, determined by the seller's mechanism's output, i.e., on making bid $b_t=f_t(s_t, v_t)$ in round $t$, the reward is given by buyer's utility 
\[u_t(b_t) := \Ex[v_t x_t - p_t | s_t, v_t, b_t]\] 
 For a no-regret learning buyer (refer to \eqref{eq:Regret}),{\it a behaviorally-optimal policy} is any bidding strategy such that for the  trajectories of bids, states and valuations $s_1, \val_1, b_1,\ldots, s_T, \val_T, b_T$, generated by this policy, we have in hindsight, 
\begin{equation}
\label{eq:no-regret}
\text{Regret}(T) = \max_{f\in {\cal E}} \sum_{t=1}^T u_t(f(s_t,v_t)) - \sum_{t=1}^T u_t(b_t) = o(T)
\end{equation}
Here, we slightly abused the notation to define $f$ as a function of $s_t,v_t$, where as technically it is only a function of $f(\indi(s_t=\bot), v_t)$, that is, it only uses whether $s_t=\bot$ or $s_t\ne \bot$.
\item {\it No-policy-regret learner:} This more sophisticated buyer uses a no-{\it policy-regret} learning algorithm. Following \eqref{eq:policyRegret1}, the important distinction from the definition of regret in the previous paragraph is that now the total utility of best expert must be evaluated over the {\it trajectory of states achieved by the expert}. To make explicit the dependence of $t^{th}$ round utility on past decisions through the state at time $t$, let us denote the utility in round $t$ as $u(b_t, s_t)$  
Then, following \eqref{eq:policyRegret1}, policy-regret of such a buyer is given by:
\begin{equation}
\label{eq:buyer-policy-regret}
\text{Policy-Regret}(T) = \max_{f\in {\cal E}} \sum_{t=1}^T u_t(f(s'_t,v_t), s'_t) - \sum_{t=1}^T u_t(b_t, s_t), 
\end{equation}
where $s'_1, \ldots, s'_T$ is the (possibly randomized) trajectory  of states that would be observed on using the expert to decide the bids in all rounds. For a no-policy-regret learning buyer, a behaviorally optimal policy is any bidding strategy such that for resulting trajectory of bids and states $s_1, b_1, \ldots, s_T b_T$, the above quantity is guaranteed to be $o(T)$.

 While constructing such a no-policy-regret learner is difficult in general, for our proposed stochastic state-based mechanisms and the above special case of single expert sequences, this is achievable by some simple learning strategies. In fact, a simple buyer learning strategy that will work for our mechanism to achieve no-policy-regret with high probability, is to explore each possible bid for some time and then use the best single bid for the rest of the time steps. 
\end{itemize}

\subsection{Desirable properties of a mechanism}
\label{sec:prelim-ir}
\label{sec:mechProp}
\paragraph{Non-payment forceful.}  A mechanism is {\em non-payment forceful} if the payment $\price_t$ is always non-negative, and $0$ when the bid is $0$, i.e.,
\begin{equation}
\label{eq:non-payment-forceful}
\price(\newstate_t,\bid_t,\alloc_t) \ge 0, \text{ and } \price(\newstate_t, 0,\alloc_t) = 0, \text{ for any } \newstate_t, \bid_t, \alloc_t.
\end{equation}
Such a mechanism has the desirable property that it cannot be forced to pay the buyer in any state and it cannot force payments out of buyers who bid $0$.
\paragraph{Interim Individually Rational (IIR)}  A mechanism is defined to be {\em interim individually rational (IIR)} 
iff at any time $t$, given any history and current valuation $\val_t$, there exists a bid $b_t$ such that the buyer's expected perceived utility is non-negative. For a $k$-lookahead buyer, this means that given any state $\newstate_t$, valuation $\val_t$, for any bid $b_t$ that is $k$-lookahead optimal at time $t$, we have 
 $$U_{\min\{k, T-t\}}^{t}(\newstate_t,\val_t,b_t) \geq 0.$$ 
Note that since the buyer's utility is a non-negative aggregate (possibly over multiple future time steps) of per time-step buyer utility, the non-payment-forceful condition
guarantees IIR.

An interim-individually rational mechanism allows take-it-or-leave-it offer based mechanisms like the following: ``pay $E[v]$ today to get the item tomorrow." A forward-looking buyer would find this offer attractive, hence the seller would extract the full
surplus, $\Ex[v]$. An unsettling feature of the afore-described mechanism is that, for some realizations of $v$, the
buyer ends up with negative utility. In particular, while this mechanism is interim Individually Rational (IR), it is not `ex-post IR'.
\paragraph{Ex-post IR}
This requires that the total utility of a rational buyer at the end of $T$ rounds is non-negative in hindsight.
 That is, there exists some behaviorally optimal policy, such that all trajectories of bids $b_t, t=1,\ldots, T$ (and implicitly the trajectories of states and valuations $s_1,\ldots, s_T$, $\val_1, \ldots, \val_T$) generated by this policy satisfy
$$\textstyle\sum_{t=1}^T u_t(b_t) \ge 0.$$ 
Here, $u_t(b_t) = \Ex_{\alloc_t \sim \alloc(\newstate_t,\bid_t)}[\alloc_t \cdot \val_t - \price_t(\newstate_t,\bid_t,\alloc_t) | \val_t, \newstate_t, b_t]$ is the round $t$ utility defined in the previous section as the buyer's value from that round's allocation minus the payment. 
\paragraph{Per-round Ex-post IR}
We also consider a stronger version of ex-post IR considered in some previous works \cite{ADH16, MPTZ16b}
: at each round, the utility of the buyer, defined to be his value from that period's allocation
minus the buyer's payment, must be non-negative. 
More precisely, we define a mechanism to be {\em per round ex-post individually rational}  
if there exists some behaviorally optimal policy such that every trajectory of bids $b_1, \ldots, b_T$ (and implicitly states $s_1,\ldots, s_T$ and valuations $\val_1, \ldots, \val_T$) generated by this policy satisfies
 \begin{align}
\forall t, u_t(b_t)  
\geq 0. \label{eq:ex post IR}
\end{align}
Per-round ex-post IR is a stronger requirement than ex-post IR, and clearly implies ex-post IR.  It is also referred to as stage-wise ex-post IR \cite{MPTZ16b}. 


\section{Our Mechanism}
\label{sec:our-mech}
We prove our positive results by proposing a simple state-based mechanism $\mthree$ parameterized by three parameters $\price$ (threshold price), $\epsilon$, and $\rho$. A high level description of the mechanism is as follows. The mechanism uses roughly the average of buyer's past accepted bids to decide the buyer's state. The mechanism can be in two types of states: `good state' if the current average is above $(1-\epsilon)\mu$,  and `bad state' otherwise. In a good state, the mechanism always accepts the buyer's bid (irrespective of the bid value), with payment equal to bid. The state transitions from a `good state' to a `bad state' if the average of accepted bids falls below $(1-\epsilon)\mu$. In a bad state, if the buyer's bid is above the threshold $\price$, then with probability $\rho$, the bid is accepted and the buyer is transferred to a good state. Any bid below $\price$ is rejected in a bad state.

Below are the precise definitions.
\begin{definition}[Mechanism $\mthree$] 

\begin{itemize}
\item {\bf State Space $\newstatespace$:} The state space is $\newstatespace=\R \times \mathbb{N}$. We represent a state $s \in \newstatespace$ by pair $(\avgbid,n)\in \R \times \mathbb{N}$, where $\bar{b}$ represents an average of $n$ past bids. We refer to states $s=(\avgbid, n)$ with 
$$\textstyle \avgbid \ge (1-\epsilon)\mu$$ 
as `good states', and all the other states as `bad states'. Abusing notation a little, if $s$ is a bad state we say $s = \bot$, otherwise, we say $s \neq \bot$. Further, we refer to any state of form $s_t = ((1-\epsilon)\mu, 0)$ as a `borderline' good state.
\item {\bf Starting state $s_1$:} The mechanism starts in a borderline good state, i.e., 
$$s_1=((1-\epsilon)\mu, 0).$$
\item  {\bf Allocation rule $x(s_t, b_t)$:} Given current state $s_t$, and bid $b_t$, this mechanism always allocates in a good state. In a bad state, it allocates with  probability $\rho$ {\it if} the bid $b_t$ is above the price $\price$. That is, $x_t \sim x(s_t, b_t)$, where
$$x(s_t, b_t)= 
\left\{ \begin{array}{ll}
1, & \text{ if } s_t \ne \bot,\\
\text{Bernoulli}(\rho), & \text{ if } s_t = \bot \text{ and } b_t \ge \price,\\
0, & \text{ otherwise.}
\end{array}\right.
$$
\item {\bf Payment rule $p(s_t, b_t, x_t)$:} This is a first price mechanism, i.e.,
$$p(s_t, b_t, x_t) = 
\left\{ \begin{array}{ll}
b_t, & \text{ if } x_t =1,\\
0, & \text{otherwise.}
\end{array}\right.
$$
Note that the payment is always smaller than bid, with $0$ payment no allocation. By definition, this mechanism is non-payment forceful in all states.
\item {\bf State-transition function $Q(s_t, b_t, x_t, p_t)$:} $Q(s_t, b_t, x_t, p_t)$ provides the distribution of next state $s_{t+1}$. Let $s_t=(\avgbid, n)$. In this mechanism, the state effectively remains the same if $x_t=0$. Otherwise, it transitions either (from good state) to a state with updated average bid, or (from bad state) to a borderline state.
$$s_{t+1}= 
\left\{ \begin{array}{ll}
(\avgbid, n), & \text{ if } x_t =0\\
\left(\frac{\avgbid n+b_t}{n+1}, n+1\right), & \text{ if } s_t \ne \bot, x_t=1\\
((1-\epsilon)\mu, 0), & \text{ if } s_t = \bot,  x_t=1,\\
\end{array}\right.
$$
\end{itemize}
\end{definition}

\section{Revenue tradeoffs} \label{sec:revenue tradeoffs}
In this section, we prove revenue guarantees for mechanism $\mthree$ against buyers with different lookahead attitudes and learning behaviors. Specifically, we demonstrate that with appropriate parameter settings, this mechanism achieves the results stated in Theorem \ref{th:1} and Theorem \ref{th:2}. 

In below, $\mpr := \arg \max_{\price}\price(1-F(\price))$ and $\revmye := \mpr(1-F(\mpr))$ denote the Myerson price and Myerson optimal revenue respectively. And, $\rev^{\mthree}$ denotes the expected value of average revenue of mechanism $\mthree$ over $T$ rounds, i.e., $\rev^{\mthree} := \frac{1}{T} \sum_{t=1}^T \price_t$. Order notation $o(1)$ and $o(T)$ will be used denote asymptotic order with respect to $T$.

\subsection{Revenue against diverse lookahead attitudes}
Theorem \ref{th:1} can be obtained as a corollary of the following proposition by substituting $\rho =\frac{\epsilon}{2-\epsilon}$.
\begin{proposition}
\label{pp:lookahead}
The mechanism $\mthree$ with parameters $p=\mpr, \rho\le \frac{\epsilon}{2-\epsilon}$ 
achieves 
\begin{enumerate}
\item[(a)] revenue $\rev^{\mthree} \ge \frac{\rho}{\rho+1} \revmye - o(1)$ against myopic buyers, while achieving 
\item[(b)] revenue $\rev^{\mthree} \ge (1-\epsilon)\mu - o(1)$ against $k$-lookahead buyers for any $k\ge 1$. 
\end{enumerate}
\end{proposition}
\begin{proofsketch}
Here, we provide a proof sketch. A detailed proof is provided in Appendix \ref{app:lookahead}.

\noindent (a) {\it Revenue against myopic buyers.} 
Let us first consider the seller's expected revenue from a myopic buyer over $T$ rounds. By definition, the mechanism starts in the borderline good state $s_0 = ((1-\epsilon)\mu, 0) \ne \bot$. Since a myopic buyer is only concerned with immediate utility, she will bid $0$ in this state to get allocation with maximum possible utility. However, this will take the average of bids below the boundary of $(1-\epsilon)\mu$ and the bidder will immediately go to the bad state. In a bad state, such a buyer will bid $\mpr$ whenever $\val_t \ge \mpr$. Therefore, the buyer will return back to the borderline good state $s_0$ with probability $\rho \Pr(\val_t \ge \mpr) = \rho (1-F(\mpr))$. Again, in $s_0$, the myopic bidder will bid $0$ and immediately transfer back to a bad state. Therefore, the buyer will roughly spend $\ell=\frac{1}{\rho (1-F(\mpr))}$ time steps in bad state for every visit to a good state. The expected number of such good state-bad state visit cycles is roughly $T/(\ell+1)$, with seller's revenue of $\mpr$ (from bad states) in every cycle. This givens the expected revenue in $T$ rounds as roughly $\frac{T\mpr}{\ell+1} \ge \frac{\rho}{\rho+1} \mpr (1-F(\mpr)) = \alpha \revmye$. An extra $-\frac{1}{T} = -o(1)$ term is obtained due to possible interruption of the last cycle.\\
(b) {\it Revenue against $k$-lookahead buyers.} For $k$-lookahead buyers for any $k\ge 1$, the key is to demonstrate that in mechanism $\mthree$ with $\rho<\frac{\epsilon}{2}$, an optimal $k$-lookahead bid in a good state guarantees that the next state is also a good state. This is proven in Lemma \ref{lem:alwaysgood}. The bound on revenue can then be obtained by observing that starting in a good state, a $k$-lookahead buyer will remain continuously in good states. And, since the bidder always gets charged her bid in a good state, this implies the seller's average revenue is $\bar{b}_T$ which by definition of good states is at least $\mu(1-\epsilon)$. Here $\bar{b}_T$ denotes the average of bids $b_1, \ldots, b_T$. 
\end{proofsketch}

The following lemma forms the key to main technical results in this paper.
\begin{lemma}
\label{lem:alwaysgood}
Assume $\rho\le \frac{\epsilon}{2-\epsilon}$. Then, at any time $t$, given that $s_t \ne \bot$,  and $b_t$ is a $k$-lookahead optimal bid, we have that $s_{t+1} \ne \bot$.
\end{lemma}
\begin{proofsketch}
To develop an intuition, let us first consider a simpler version of our mechanism where at every step in the bad state, there is an independent probability $\rho$ of getting an allocation and transferring back to the good state, irrespective of the bid. (Note that this mechanism, that does not charge anything in bad states, will result in $0$ seller's revenue against myopic buyers). 

Assume that current state $s_t$ is a good state. Suppose for contradiction that $b$ is an optimal bid at time $t$ such that $s_{t+1}= \bot$. Let us compare the $k$-lookahead utility of this bid to a bid $b' \le (1-\epsilon)\mu$ that keeps the bidder in a  good state $s_{t+1}$ (note that such a bid always exists). The comparison for $1$-lookahead buyer is relatively simple. In both cases, such a buyer will plan to bid $0$ in the next round (since next round is the last round in $1$-lookahead buyer's planning horizon). When $s_{t+1}\ne \bot$, $x_{t+1}=1$, and therefore, $1$-lookahead utility for $b'$ is at least
\begin{center} $v_t-(1-\epsilon)\mu_t + \Ex[\val_{t+1}] = v_t+\epsilon\mu$. \end{center}
In comparison, when $s_{t+1}=\bot$, the probability of allocation is $\rho$, and therefore $1$-lookahead utility for $b$ is  at most 
\begin{center} $v_t + \rho \Ex[\val_{t+1}] = v_t+\rho\mu$. \end{center}
Therefore, if $\rho<\epsilon$, $b'$ is clearly a better choice than the optimal bid $b$, thus giving a contradiction. For $k$-lookahead buyers, extending this argument involves coupling the $k$-length trajectories of two bidding strategies: one that plays optimal bids starting from bid $b$ in $s_t$ vs. a strategy that starts from bid $b'$ in state $s_t$. By coupling these two trajectories, we aim to compare the $k$-lookahead utilities of the two bids, in order to demonstrate that playing $b$ in the state $s_t$ is strictly suboptimal. 

 The coupling point of the two trajectories is the (random) time step $t+\tau+1$ where the (supposedly) optimal strategy's trajectory (that started with $b$) comes back to a (borderline) good state for the first time. Define the second bidding strategy in a coupled manner to play bid $b'=(1-\epsilon)\mu$ repeatedly until $t+\tau$, and optimal bids thereafter. Then, at step $t+\tau+1$, the second trajectory is also in a good state, but possibly not borderline. To relate the utilities of the two trajectories after this coupling point, we use induction along with an observation that the optimal utility for a trajectory starting at any good state is at least as good as the utility for a trajectory starting at a borderline good state (refer to Claim \ref{claim:border} in appendix). Therefore, from point $t+\tau+1$ onwards, the second trajectory has better, or at least as good utility, as the first supposedly optimal strategy. It remains to compare the utilities in steps $t, t+1, \ldots, t+\tau$. 

By definition of $\tau$, under first supposedly optimal strategy, the states at time steps $t+1, \ldots, t+\tau$ are all bad states.  Therefore, the total expected utility is at most
\begin{center}$\val_t + \tau \rho \mu$ \end{center}
On the other hand, the second strategy of bidding $(1-\epsilon)\mu$ always stays in good state to get total expected utility of at least $\val_t -(1-\epsilon)\mu + \tau \epsilon \mu$ in these steps.
A better bound for the second case can be observed by considering the probability that the `last' (i.e., $k^{th}$) round occurs at the end of this interval, that is, $\tau=k$. At the last round, the second strategy bids $0$ to obtain full valuation as utility. Therefore, a better lower bound for expected utility is 
\begin{center} $\val_t -(1-\epsilon)\mu + (\tau-1)\epsilon \mu + \Pr(\tau=k) \mu$ \end{center}
Here $\Pr(\tau=k) = (1-\rho)^{k-1}\rho$. Comparing this with the utility of the first strategy, after some careful algebraic manipulations we obtain that under condition $\rho\le \epsilon/(2-\epsilon)$, the utility for second strategy is better than optimal; this gives a contradiction. 

Now, lets consider the more complicated allocation rule actually used by $\mthree$ in bad state: the allocation in bad state actually depends on the bid, and the bidder gets an allocation with probability $\rho$ only if the bid is above $p$. Although this makes the payoff for bidder in bad states worse, and provides more incentive to remain in good state; it also makes analysis trickier since now the event of a bidder transferring from bad state to good state depends on her bid (and therefore her valuation) in addition to the random $\rho$ probability event. 	

Finally, all the coupling arguments and expected revenue analysis sketched above need to be carried out carefully using martingale analysis and stopping times. The details are provided in Appendix \ref{app:lookahead}.
\end{proofsketch}

\subsection{Revenue against buyer learning behaviors}
Theorem \ref{th:2} can be obtained as a corollary of the following proposition, by substituting $\rho=\epsilon - o(1)$.
\begin{proposition}
\label{pp:learning}
The mechanism $\mthree$ with $\rho < \epsilon$ and $p=\mpr$ achieves a revenue of at least
\begin{enumerate}
\item[(a)] $\rev^{\mthree} \ge \frac{\rho}{\rho+1} \revmye - o(1)$  against any 
buyer who is a no-regret learner for the class ${\cal E}$ of experts (refer to 
Equation~\eqref{eq:experts} and~\eqref{eq:no-regret}), and 
\item[(b)] $\rev^{\mthree} \ge (1-\epsilon)\mu - o(1)$ against any buyer who is a policy regret learner for a class ${\cal C}$ of sequences containing all sequences of single experts (refer to Equation \eqref{eq:buyer-policy-regret}).
\end{enumerate}
\end{proposition}
\begin{proof} 

\noindent (a) {\bf No-regret learners.} Consider a bidding function $f(s_t, v_t)$ defined as $f(s_t, v_t)=p^*$ when $s_t = \bot$ and $v_t\ge \mpr$, and $0$ otherwise. This is (arbitrarily close to) one of the experts in the class ${\cal E}$ of experts that the buyer is using.  In the mechanism $\mthree$, with probability $\rho$, the bid $\mpr$ made in a bad state will get accepted, to earn utility $u_t(f(s_t, v_t)) = \rho(v_t-\mpr)$ for the buyer. Therefore, using the above expert as benchmark, we obtain that for any sequence of states and valuations, the first term in the regret definition \eqref{eq:no-regret} is at least 
$$\textstyle \sum_{t=1}^T u_t(f(s_t,v_t)) \ge  \rho \sum_{t:s_t = \bot, v_t\ge \mpr} (v_t -\mpr)  + \sum_{t:s_t\ne \bot}v_t$$
Since the buyer is using a no-regret learning algorithm, she must be achieving a utility that is within $o(T)$ of the above utility. Now, the maximum utility achievable in any good state is $\val_t$. And, the buyer cannot make any positive utility in a bad state at time $t$ if $\val_t \le \mpr$. Therefore, a no-regret learning buyer cannot afford to lose more than $o(T)$ of the bad state auctions with $v_t > \mpr$. 
This means that any no-regret learning buyer must bid $b_t \ge \mpr$ in all but possibly $o(T)$ of the rounds $t$ with $s_t= \bot, \val_t\ge \mpr$. Let $B$ be a random  variable denoting the number of bad states in the state trajectory, and let $B'$ denote the number of those bad states with $v_t\ge \mpr, b_t \ge \mpr$, then,
$$\Ex[B'] \ge \Ex[B](1-F(\mpr))-o(T)$$
Also, let $G$ be the number of good states. 

To lower bound $\Ex[B]$, let us first consider state trajectories of form 
\begin{center} \mbox{$\not\perp\bot\bot\bot\bot\bot\not\perp\bot\bot\bot\not\perp\bot\bot\ldots$},\end{center}
 i.e., lone good states interspersed with sequences of bad states. 
 Then, due to the construction of the mechanism $\mthree$, the only way the buyer can obtain $G-1$ good states (all except the first good state), only by winning $G-1$ bad state auctions. And, since the bad state auctions can be won only with probability $\rho$, we have that $\Ex[G] \le \Ex[B]\rho + 1$. This gives,
$$\textstyle \Ex[B] \ge \frac{1}{\rho+1} (T-1),$$
Combining the above observations,
$$ \textstyle \Ex[B'] \ge \Ex[B](1-F(\mpr)) - o(T)\ge \frac{1}{\rho+1} (1-F(\mpr)) (T-1) -o(T).$$
Therefore, the total expected seller's revenue is at least 
$$\rho \mpr \Ex[B'] \ge \frac{\rho}{\rho+1} \mpr (1-F(\mpr)) (T-1) -o(T)= \frac{\rho}{\rho+1}  \revmye T -o(T)$$
Now, consider sequences of states with more than one consecutive good states, e.g., 
\begin{center}
$\not\perp \bot\bot\bot\bot \not\perp\not\perp\not\perp \bot\bot\bot\bot \not\perp\not\perp\ldots$ etc. 
\end{center}
Then, in any sub-sequence of consecutive good states, the bid average over all good states except the first one (call them trailing states) must be at least $(1-\epsilon)\mu$, so that the buyer makes at most $\bar{v}-(1-\epsilon) \mu$ utility on average, where $\bar{v}$ denotes the average valuation over the trailing states. On the other hand, the above expert $f$, which bids $0$ in good states, makes an average of $\bar{v}$ utility in those trailing states (in hindsight). Further, in the bad states, and in the rest of the (non-trailing) good states, $f$ is achieving the best possible utility. Therefore, given the no-regret condition, the number of trailing states can be at most $o(T)$ and do not effect the revenue calculations above.

\noindent (b) {\bf No-policy-regret learners.}
Here, we use the sequence of constant bids $(1-\epsilon)\mu$ as a benchmark expert. This bidding strategy ensures that the mechanism is always in a good state $s_t = ((1-\epsilon)\mu, t)$, and achieves a utility of $\epsilon \mu$. Therefore, since the class ${\cal C}$ contains such sequences of single experts, the policy-regret learning buyer must achieve at least $\epsilon \mu -o(T)$ utility. Now, in bad states, the buyer can achieve at most $\rho \mu$ utility on average. Therefore, if $\rho < \epsilon$, then the number of bad states in buyer's state trajectory can be at most $o(T)$. This implies that the trailing good states are at least $T-o(T)$. By definition, the bidding average over the trailing good states must be least $(1-\epsilon)\mu$, and therefore, the seller's revenue is at least $(1-\epsilon)\mu T -o(T)$.
\end{proof}

\section{Other related work}
There are several streams of literature in dynamic mechanism design. We begin 
with the stream that is closest to our work. 

{\noindent \bf  Optimal dynamic mechanisms.} \cite{PPPR16} show that the optimal deterministic dynamic mechanism satisfying
ex-post IR constraints even in a single buyer 2 rounds setting, when the values
are correlated is NP-hard. I.e., buyer learns his value of each round when it
begins, and both buyer and seller know the distribution from which these values
are drawn. They show that the optimal deterministic mechanism when the rounds
have independent valuations can be computed in polynomial time. The optimal
randomized mechanism even with correlated valuations can be computed in
polynomial time.~\cite{MPTZ16a} study the single seller single buyer setting and
show that with the IIR constraint, a very simple class of mechanisms called
bank-account mechanisms that maintain a single scalar variable as state already
obtain a significantly higher revenue and welfare compared to the single shot
optimal.~\cite{ADH16} and~\cite{MPTZ16b} characterize the optimal ex-post
IR mechanisms and consider approximations thereof via simple mechanisms (mainly in the single seller, single buyer
setting, but their results also extend to the multi-bidder case) that again hold
a single scalar variable as state. ~\cite{BMP16}, consider a single seller
single buyer setting and show that the seller can earn almost the entire surplus
as revenue, even after imposing per round ex-post IR requirements and martingale
utilities for the buyer.~\cite{MPTZ17} study oblivious dynamic mechanism design,
namely, one where the seller is not aware of the future distributions of the
buyer, and just the distribution for this round: they show that even with just
this information, one can construct an ex-post IR dynamic incentive compatible
mechanism that gets a $\frac{1}{5}$ of the optimal dynamic mechanism that knows
all future distributions.

{\noindent \bf Mechanism design for buyers with evolving values.} Another major 
focus area in dynamic mechanism design is one where buyers experience the same 
or related good repeatedly over time, and their value for the good evolves with 
time/usage. Initiated by the work of Baron and Besanko~\cite{BB84} there is a 
large body of work~\cite{Besanko85,Battaglini05,CH00,BP14,ES07} that study 
optimizations in the presence of evolving values. Recent works include those 
by~\cite{BS15,AS13}, where they consider general models where value evolution 
could depend on the action of the mechanism.~\cite{KLN13,PST14} study revenue 
optimal dynamic mechanism design where the buyer's value evolves based on 
signals that she receives each period.~\cite{CDKS16} study  martingale 
value evolution for the buyer and show that simple constant pricing schemes 
followed by a free trial earns a constant fraction of the entire surplus.  

{\noindent \bf Bargaining, durable goods monopolist and Coase conjecture.} There 
is a large body of literature in economics that studies settings where the value 
is initially drawn from a distribution, but in subsequent rounds, the value 
remains the same, i.e., there is not a fresh draw in every round. This setting 
can be motivated based on several applications including bargaining, durable 
goods monopoly and behavior based discrimination. See~\cite{FV06} for an 
excellent survey and references there in for an overview of this area. 

{\noindent \bf Dynamically arriving and departing agents.} Yet another body of work 
that comes under the umbrella of dynamic mechanisms is one where agents arrive 
and depart dynamically. Naturally focus is quite different from what we do in 
this paper.

{\noindent \bf Lookahead Search.}
The study of $k$-lookahead search can be viewed in the context 
of {\em bounded rationality}, as pioneered by Herb Simon \cite{Sim55}.  He argued that,
 instead of optimizing, agents may apply a class of heuristics, termed satisfying heuristics in
decision making, A natural choice of such heuristics is restricting
the search space of best-response moves.
Lookahead search in decision-making has been motivated and examined in great extent by the artificial intelligence community~\cite{N83,KRS92,SKN09}. Lookahead search is also related to the sequential thinking framework in game theory \cite{SW94}.  More recently, ~\cite{MTV12} study the quality of equilibrium outcomes for look-ahead search strategies for various classes of games.
They observe that the quality of resulting equilibria increases in generalized second-price auctions, and duopoly games, but not in other classes of games. No prior work studies dynamic mechanisms that are robust against various lookahead search strategies.

\bibliographystyle{plainnat}
\bibliography{klookahead}

\appendix
\section{Mechanism for per-round ex-post IR with infinite lookahead buyers}
\label{sec:inf-lookahead}
As discussed in the introduction, it is not possible to get any 
reasonable approximation to $\mu$ as revenue against $k$-lookahead buyers for 
small $k$ when per-round ex-post IR is imposed. Nevertheless, we show here how 
to make small modifications to the mechanism discussed in 
Section~\ref{sec:our-mech} to get per-round ex-post IR in all but a constant 
number of rounds for an infinite lookahead buyer (i.e., a $T$-lookahead buyer 
in a $T$-rounds game).

Let $B$ be the highest value 
in the support of the buyer value distribution and $\mu$ the mean of the 
distribution.  The target-revenue $R_{target}$ be defined as

\begin{align}
\label{eqn:r-target}
R_{target} &= T\mu(1-\epsilon)
- \sqrt{4B\mu\sqrt{T}\ln(T)} - 
\sqrt{2B\mu\ln(T)}\sum_{j=1}^T\frac{1}{\sqrt{j}}\\
&= T\mu(1-\epsilon) - c\cdot\sqrt{B\mu T\ln(T)}
\end{align}

\begin{definition}[Mechanism $M_\infty(\epsilon,\rho,p)$]
\begin{enumerate}
\item {\bf State Space $\newstatespace$:} The state space is $\newstatespace=\R 
\times \R \times \mathbb{N}$. A state $s = (TP,EP, t)$ is a good 
state if $TP \geq EP$, and bad otherwise. $TP$ stands for the total 
payment made by the buyer in the good state, and $EP$ stands for the payment 
the mechanism expects the buyer to have paid and $t$ is the total number of 
rounds elapsed, including the current round that's beginning. 

\item {\bf Starting state $s_1$:} The mechanism starts at $t=0$, $EP = 0$, and 
$TP = \mu\sqrt{T} + \sqrt{4B\mu\sqrt{T}\ln(T)}$. I.e., the mechanism the buyer 
a credit of $\mu\sqrt{T} + \sqrt{4B\mu\sqrt{T}\ln(T)}$ to begin with.
\item  {\bf Allocation rule $x(s_t, b_t)$:} Given current state $s_t$, and bid 
$b_t$, this mechanism always allocates in a good state. In a bad state, it 
allocates with  probability $\rho$ {\it if} the bid $b_t$ is above the price 
$\price$. That is, $x_t \sim x(s_t, b_t)$, where
$$x(s_t, b_t)= 
\left\{ \begin{array}{ll}
1, & \text{ if } TP \geq EP,\\
\text{Bernoulli}(\rho), & \text{ if } TP < EP \text{ and } b_t \ge \price,\\
0, & \text{ otherwise.}
\end{array}\right.
$$
\item {\bf Payment rule $p(s_t, b_t, x_t)$:} This is a first price mechanism, 
with a small change:
$$p(s_t, b_t, x_t) = 
\left\{ \begin{array}{ll}
\min(b_t, R_{target} - TP) & \text{ if } x_t =1 \text{ and } TP \leq 
R_{target},\\
0, & \text{otherwise.}
\end{array}\right.
$$
\item {\bf State-transition function $Q(s_t, b_t, x_t, p_t)$:} $Q(s_t, b_t, 
x_t, p_t)$ provides the distribution of next state $s_{t+1}$. Let 
$s_t=(TP, EP, t)$. In this mechanism, the state effectively remains the 
same if $x_t=0$. Otherwise, it transitions to a state with updated total 
payment and expected payment. The total payment naturally increases by $p_t$. 
The expected payment increases not by $\mu(1-\epsilon)$, but has an additional 
slack, namely, it increases by $\mu(1-\epsilon) - \sqrt{\frac{2B\mu\ln(T)}{t}}$.
$$s_{t+1}= 
\left\{ \begin{array}{ll}
(TP, 0, t+1), & \text{ if } TP \geq R_{target}\\
(TP+p_t, EP+\mu(1-\epsilon) - \sqrt{\frac{2B\mu\ln(T)}{t}}, t+1), & \text{ if } 
R_{target} > TP \geq 
EP\\
(TP, EP, t+1), & \text{ if } TP < EP, x_t =0\\
(TP, TP, t+1), & \text{ if } TP < EP,  x_t=1,\\
\end{array}\right.
$$
\end{enumerate}
\end{definition}

The main difference between this mechanism $M_\infty(\epsilon,\rho,p)$ and the 
$\mthree$ defined 
in 
Section~\ref{sec:our-mech} is that when in the boundary state of $TP = EP$ at 
the beginning of round $t$, in $M_\infty$ the agent needs to bid 
$\mu(1-\epsilon) - 
\sqrt{\frac{2B\mu\ln(T)}{t}}$ to stay in a good state, as opposed to 
$\mu(1-\epsilon)$ in $\mthree$. I.e., there is a slack that vanishes with time, 
i.e., as $t$ 
gets large. Thus having to bid even lower than what he had to in 
Section~\ref{sec:our-mech} to remain in the good state, it immediately follows 
that Lemma~\ref{lem:alwaysgood} holds 
for our modified mechanism as well, i.e., any agent in a good state, will 
continue to be in a good 
state as it is less costly. This immediately implies identical revenue 
guarantees for all lookahead buyers, with a $c\mu\sqrt{T}$ loss in revenue due 
to the vanishing slack this mechanism provides in every round (and also the 
slack initially). 

\paragraph{Truthful bidding will not lead a bidder out of 
good state w.h.p.} Consider a bidder bidding his true 
value $V_t$ in round $t$. The probability that, for any $r$, after $r$ rounds 
of bidding true value $V_t$ we have $TP < EP$ is at most $\frac{1}{T^c}$ by 
Chernoff bounds (the slacks are chosen in $M_\infty$ to satisfy this). 
Therefore, by union bound, the probability that it ever 
happens in any of $T$ rounds is at most $\frac{T}{T^c} \rightarrow 0$ as $T \to 
\infty$. I.e., with a high probability, truthful bidding will never lead a 
buyer out of good state. 

\paragraph{Truthful bidding is optimal w.h.p.} Note that the mechanism doesn't 
accept payment once it has earned a revenue of $R_{target}$. This $R_{target}$ 
is computed to be the minimum revenue that the mechanism is guaranteed to 
extract from a buyer that always lives at the boundary state. Or 
equivalently $T\mu - R_{target}$ is the maximum utility such a buyer can earn, 
and this in turn implies that this is the maximum utility an infinite lookahead 
buyer can earn. To achieve this utility, an infinite lookahead buyer need not 
shade his bid to constantly live at the boundary state --- just bidding his 
true value is an equally good strategy (modulo the tiny probability that this 
will get him to a bad state, where he bids whatever is necessary to maintain 
him in good state) in terms of utility. Thus, bidding the true value throughout 
all the $T$ rounds (which is clearly a per-round ex-post IR strategy)
is an optimal strategy with probability $1-\frac{1}{poly(T)}$. 

The above discussion gives us the proof of Theorem~\ref{thm:expost-ir-inf}, 
restated below
\begin{oneshot}{Theorem~\ref{thm:expost-ir-inf}}
For any $\epsilon\in (0,1)$, there exists a non-payment forceful and IIR 
mechanism that is $(\frac{\epsilon}{2}, 1-\epsilon)$-robust against {\it 
lookahead attitudes}, ex-post IR against a myopic buyer, and with a high 
probability ex-post IR against an 
infinite lookahead buyer. 
\end{oneshot}

\section{Missing Proofs}
\label{app:lookahead}
\paragraph{Proof of Proposition~\ref{pp:lookahead}}	
	 Here, we provide a detailed proof of Proposition \ref{pp:lookahead}. Lemma \ref{lem:Mthree0} proves the result for myopic buyers. Next, Lemma \ref{lem:Mthree1} and Lemma \ref{lem:transitionGood} together provide a detailed proof for  Lemma \ref{lem:alwaysgood}, to complete the proof of Proposition \ref{pp:lookahead}.

\begin{lemma}
\label{lem:Mthree0}
The mechanism $\mthree$ with $p=\mpr$ satisfies the following properties against a myopic buyer:
\begin{itemize}
\item[(a)] An optimal myopic bid exists for all time steps $t$. 
\item[(b)] The average per round revenue against a myopic buyer is bounded as:
$$\revmthreemyek[0]  \geq 
\frac{\rho}{\rho+1} R_\mye -\frac{1}{T},$$ where $R_\mye$ is the one round Myerson revenue, 
namely, $R_\mye=\max_{\price}\price(1-F(\price)) = \mpr(1-F(\mpr))$. 
\end{itemize}
\end{lemma}
\begin{proof}

\noindent {\it Buyer's optimal myopic bid:} Let us first understand a myopic buyer's bidding strategy under mechanism $\mthreemye$. At any time $t$, given $s_t, \val_t$, the myopic buyer makes the bid $b_t$ that maximizes $0$-lookahead expected utility, given by
$U^t_0(s_t, \val_t, b_t)= \Ex[\alloc_t \val_t - p_t | \val_t, s_t]$. Now, consider two cases:\\
(a) $s_t\ne \bot$: in this case, by definition of allocation rule in $\mthreemye$, $\alloc_t=1, p_t=b_t$, so that $U^0_t = \Ex[\val_t - b_t]$, irrespective of the value of bid $b_t$. Therefore, the (unique) utility maximizing strategy for the buyer is to bid $0$, i.e., 
\begin{center}
$b_t=0$ when $\newstate_t\ne \bot$.
\end{center}
(b) $s_t= \bot$: in this case, by definition of allocation rule in $\mthreemye$, with probability $\rho$, $\alloc_t=1, p_t=b_t$ if $b_t\ge \mpr$, so that $U^0_t = \Ex[(\val_t - b_t)\indi(b_t \ge \mpr)].$
Therefore, the whenever $\val_t\ge \mpr$, the (unique) utility maximizing strategy for the buyer is to bid $\mpr$, otherwise $\alloc_t=0$ and any bid less than $\mpr$ (including $0$) is optimal, i.e., 
\begin{center}
$b_t=\mpr, \alloc_t=1$ when $\newstate_t = \bot, \val_t\ge \mpr$, and\\
 $b_t<\mpr, \alloc_t=0$ when $\newstate_t = \bot, \val_t< \mpr$ 
\end{center}

\noindent {\it Seller's revenue}
The seller's expected revenue is $\Ex[\frac{1}{T} \sum_{t=1}^T p_t]$. Now, for mechanism $\mthreemye$, $p_t=b_t$ whenever $\alloc_t=1$ and $0$ otherwise. For a myopic buyer as described above, $\Ex[p_t| s_t\ne \bot] =0$, $\Ex[p_t| s_t = \bot] = \mpr \Pr(\val_t \ge \mpr) = \mpr(1-F(\mpr))$. Substituting:
\begin{eqnarray*}
 \revmthreemyek[0] & = & \Ex[\frac{1}{T} \sum_{t=1}^T p_t]\\
& = & \frac{1}{T} \sum_{t=1}^T \Ex[p_t | s_t\ne \bot]\Pr(s_t \ne \bot) + \Ex[ p_t | s_t = \bot] \Pr(s_t = \bot) \\
& = & \mpr(1-F(\mpr)) \frac{1}{T} \sum_{t=1}^T \Pr(s_t = \bot).
\end{eqnarray*}
Here, $\sum_{t=1}^T \Pr(s_t = \bot)$ is the expected number of times bad state is visited in the $T$ time steps. Now, by definition, mechanism $\mthreemye$ starts in good state $s_0 = ((1-\epsilon)\mu, 0) \ne \bot$. A myopic buyer will bid $0$ in this state which will get accepted (see the above discussion in optimal myopic bid), and she will immediately go to the bad state $((1-\epsilon)\frac{\mu}{n+1}  , 0) = \bot$. Transfer from bad state to the borderline good state $s_0=((1-\epsilon)\mu, 0)$ happens with probability $\rho \Pr(\val_t \ge \mpr) = \rho (1-F(\mpr))$. Again, in $s_0$ the myopic bidder will bid $0$ and immediately transfer back to a bad state. Therefore, the sequence of states takes the form 
$\not\perp\bot\bot\bot\bot\bot\not\perp\bot\bot\bot\not\perp\bot\bot\ldots$, i.e., sequence of bad states interspersed with {\it single} good states. 
The expected length of a subsequence $\not\perp\bot^+$ is $1+\frac{1}{\rho (1-F(\mpr))}$, with $\frac{\rho (1-F(\mpr))}{1+\rho (1-F(\mpr))} \ge \frac{\rho}{\rho+1}$ fraction of bad states. Accounting for the interruption in the last $\not\perp\bot^+$ sequence due to end of time horizon $T$, we have that the expected number of steps in a bad state is at least 
$$T \frac{\rho (1-F(\mpr))}{1+\rho (1-F(\mpr))}-1 \ge T\frac{\rho}{\rho+1} -1.$$ 
Substituting, we get:\vspace{-0.1in}
\begin{eqnarray*}
 \revmthreemyek[0] & \ge & \mpr(1-F(\mpr)) \left(\frac{\rho}{\rho+1} -\frac 1 T\right)
\end{eqnarray*}
\end{proof}

\begin{lemma}
\label{lem:Mthree1}
Under mechanism $\mthree$ with $\rho \le \epsilon$, at any time $t<T$, an optimal $1$-lookahead bid exists, and is such that the next state $s_{t+1}$ is deterministically a good state, i.e., $s_{t+1}\ne \bot$.
\end{lemma}
\begin{proof}
By definition, an optimal $1$-lookahead bid $b_t$ (if exists) maximizes $1$-lookahead expected utility, i.e., 
\begin{center}
$b_t= \arg \max_b U^t_1(s_t, \val_t, b)$
\end{center}
In mechanism $\mthree$, for any $t$ such that $s_t \ne \bot$, the allocation and payment are always $x_t=1, p_t=b_t$. Therefore, using recursive relation between $U^t_k$ and $U^{t+1}_{k-1}$, 
\begin{equation}
\label{eq:1}
U^t_1(s_t, \val_t, b_t) = \Ex[ (\val_t - b_t) + \sup_{b'} U^{t+1}_0(s_{t+1}, \val_{t+1}, b') | \val_t, s_t]
\end{equation}
where 
$s_{t+1} \sim Q(s_t, b_t, 1, b_t)$. 
Now, for mechanism $\mthree$, bids get accepted in bad state with at most $\rho$ probability, therefore,
$$\Ex[\sup_{b'} U^{t+1}_0(s_{t+1}, \val_{t+1}, b')| s_{t+1}=\bot] \le \rho \mu,$$ 
where as 
$$\Ex[\sup_{b'} U^{t+1}_0(s_{t+1}, \val_{t+1}, b')| s_{t+1} \ne \bot] = \mu$$
which can be achieved by $b'=0$. Also, for $\mthree$, if $s_t$ is a good state, then depending on the bid, the next state $s_{t+1}$ will be deterministically bad or good state. For any bid $b_t$ such that $s_{t+1}$ is a bad state, substituting above in \eqref{eq:1}, we have that $1$-lookahead utility is at most
$$\val_t + \rho \mu,$$
where as if $s_{t+1}$ is a good state, then $U^t_1(s_t, \val_t, b_t) \ge (\val_t - b_t) + \mu$. 
This is maximized by the minimum bid required to keep $s_{t+1}$ as good state. 
In fact, in any good state $s_t$, the bid $b_t=(1-\epsilon)\mu$, always ensures $s_{t+1}$ is a good state, and makes the $1$-lookahead utility at least 
$$ \val_t + \epsilon\mu$$
Therefore, if $\epsilon>\rho$, then there exists at least one bid such that $s_{t+1}\ne \bot$ with strictly better $1$-lookahead utility than any other bid such that $s_{t+1}= \bot$. This proves that any $1$-lookahead optimal bid will have the stated property. 
\end{proof}

\begin{lemma}
\label{lem:transitionGood}
Under mechanism $\mthree$ with $\rho \le \frac{\epsilon}{2-\epsilon}$, for any $k\ge 1$ and time $t$, such that $s_t \ne \bot$, an optimal $k$-lookahead bid exists,
and is such that the next state $s_{t+1}$ is deterministically a good state, i.e., $s_{t+1}\ne \bot$.
\end{lemma}
\begin{proof}
\delete{ 
By definition, a $k$-lookahead optimal policy maximizes $k$-lookahead expected utility. 
Then, the bid $b_t$ given by $k$-lookahead optimal policy,
\begin{center}
$b_t \sim \pi_k^*(s_t, \val_t, t)$, where
$\pi_k^*=\arg \max_\pi U^t_{k}(s_t, \val_t, \pi)$
\end{center}
}
We prove by induction. In Lemma \ref{lem:Mthree1}, this property was proven for $1$-lookahead policy. Assume this is true for $1,\ldots,k-1$, then we prove for $k$.

By definition, a $k$-lookahead optimal bid (if exists) maximizes the $k$-lookahead utility. 
In mechanism $\mthree$, if $s_t\ne \bot$, then $x_t=1, p_t=b_t$, and depending on the value of bid $b_{t}$, the next state $s_{t+1}$ is either bad or a good state deterministically. 
Suppose for contradiction that the $k$-lookahead optimal bid $b_t$ is such that $s_{t+1}$ is a bad state. Then, we will show that there exists a bid $b'_t$ that achieves strictly better $k$-lookahead utility.

Consider the bidding strategy that bids $k$-lookahead optimal bid $b_t$ at time $t$, $k-1$-lookahead optimal bid $b_{t+1}$ at time $t+1$, and so on. 
Let $\tau \in [1,k]$ be a random variable defined as the minimum of $k$ and the number of steps it takes to reach a good state under this strategy, when starting from the bad state $s_{t+1} $ at time $t+1$, i.e., minimum $\tau$ such that $s_{t+\tau+1}\ne \bot $ or $\tau=k$.
\delete{
  For any bid $b$ such that $s_{t+1}$ is a bad state,  utility $U^t_{k}$ (refer to \eqref{eq:1} above) is at most 
$$ \val_t  + \Ex[U^{t+1}_{k-1}(s_{t+1}, \val_{t+1}, \pi^{t+1}_{k-1}) | s_{t+1}=\bot]$$
In good state $s_t$, the bid $b_t=(1-\epsilon)\mu$, always ensures $s_{t+1}$ is a good state, and the utility for this bid is at least  $U^t_{k}$ (refer to \eqref{eq:1} above) 
$$ \val_t -(1-\epsilon)\mu + \Ex[U^{t+1}_{k-1}(s_{t+1}, \val_{t+1}, \pi^{t+1}_{k-1}) | s_{t+1} \ne \bot]$$
Therefore, using Lemma \ref{lem:connectLemma} for $t+1, k-1$, the difference between the two utilities is at most
$$\Delta = (1-\epsilon)\mu  + \psi(k-1)$$
We show that $\Delta\le 0$, so that using a strategy such that $s_{t+1} \ne \bot$ has better $k$-lookahead utility for the buyer than any strategy with $s_{t+1} =\bot$.
}
Now, for $i=1,\ldots, \tau$, let $A_{t+i}$ be the event that a $\text{Bernoulli}(\rho)$ coin toss is a success. 
In $\mthree$ mechanism, in a bad state $s_{t+i}$, nothing gets added to the utility if $A_{t+i}$ is false, and at most $v_{t+i}$ gets added to the utility if $A_{t+i}$ is true. 
Therefore, the contribution to the utility in steps $t, t+1\ldots, t+\tau$ is upper bounded by 
$$\val_t + \sum_{i=1}^\tau v_{t+i} \indi(A_{t+i}).$$
Therefore,
\begin{eqnarray}
\label{eq:piOptimal}
U^{t}_k(s_{t}, \val_{t},b_t)  & \le & \val_t + \Ex[\sum_{i=1}^\tau v_{t+i} I(A_{t+i}) \nonumber\\
& & + \indi(\tau<k) \Ex[ \sup_b U^{t+\tau+1}_{k-\tau-1}(s_{t+\tau+1}, \val_{t+\tau+1}, b) | s_{t+\tau+1}]  | s_t, \val_t]
\end{eqnarray}

Next, we compare the above upper bound on utility achieved by $b_t$ to the $k$-lookahead utility achieved by $b'_t=(1-\epsilon)\mu$ at time $t$. To lower bound this $k$-lookahead utility, we consider the utility of the following bidding strategy starting from $b'_t=(1-\epsilon)\mu$ at time $t$. 
Let $\tau'$ be a random variable which given $s_t, \val_t$, has the same distribution as the random variable $\tau$ defined above. Then, 
\begin{itemize}
\item in steps $t$ to $t+\min\{\tau',k-1\}$, bid $(1-\epsilon)\mu$
\item if $\tau'=k$, bid $0$ at time $t+\tau'$
\item if $\tau'<k$, use $k-j-1$ lookahead optimal bid starting for time $t+j+1$ for $j=\tau',\tau'+1, \ldots, k-1$.
\end{itemize}
Then, the $k$-lookahead utility for bid $b'_t=(1-\epsilon)\mu$ can be lower bounded by the utility of the above strategy
\begin{eqnarray}
\label{eq:piPrime}
U^t_k(s_t, v_{t}, b'_t)
& \ge & \val_t - (1-\epsilon)\mu + \Ex[\sum_{i=1}^{\tau'} v_{t+i} - (1-\epsilon)\tau' + \indi(\tau'=k) (1-\epsilon)\mu \nonumber\\
& & \ +  \indi(\tau'<k) \Ex[ \sup_b U^{t+\tau'+1}_{k-\tau'-1}(s_{t+\tau'+1}, \val_{t+\tau'+1},b) | s_{t+\tau'+1} ]  | s_t, \val_t]
\end{eqnarray}

Now, we show that the last term from \eqref{eq:piPrime} dominates the last term from \eqref{eq:piOptimal}.
Since $\tau$ and $\tau'$ have the same distribution, in fact it suffices to compare only the expected sup utility terms for each $i$. 

Note that by definition of $\tau$, when $\tau=i$, the state $s_{t+i+1}$ reached in \eqref{eq:piOptimal} is a borderline good state, i.e.,  $s_{t+i+1} = ((1-\epsilon)\mu, n)$ for some $n$. Also, the bidding strategy used to obtain \eqref{eq:piPrime} is such that it doesn't leave  the good state until at least time $t+\tau'+1$. Therefore, when $\tau'=i$ the state $s_{t+i+1} \ne \bot$.
 Now, using Claim \ref{claim:border} (stated and proved below) for $k-i-1$, we have for all $i$, $s \ne \bot$,
\begin{equation}
\label{eq:tmp}
\Ex[\sup_b U^{t+i+1}_{k-i-1}(s_{t+i+1}, \val_{t+i+1}, b | s_{t+i+1} = s] \ge \Ex[\sup_b U^{t+i+1}_{k-i-1}(s_{t+i+1}, \val_{t+i+1}, b | s_{t+i+1} \isborder]
\end{equation}
Therefore, we derive that the last term  in \eqref{eq:piPrime},  is greater than or equal to the corresponding term
in~\eqref{eq:piOptimal}.

Using this observation, and subtracting \eqref{eq:piOptimal} from \eqref{eq:piPrime}, we can bound the total difference (denoted as $\Delta$) in $k$ lookahead utilities of $b_t$ and $b'_t$ as 
\begin{eqnarray*}
\Delta & := &  \textstyle U^t_k(s_t, v_{t}, b'_t) - U^t_k(s_t, v_{t}, b_t) \\
& \ge & \textstyle \Ex[\sum_{i=1}^{\tau'} v_{t+i} - (\tau'+1) (1-\epsilon)\mu + I(\tau'=k) (1-\epsilon)\mu | s_t, \val_t]- \Ex[\sum_{i=1}^{\tau} v_{t+i} I(A_{t+i}) | s_t, \val_t]
\end{eqnarray*}
Since $\tau$ and $\tau'$ have the same distribution given $s_t, \val_t$, we can replace $\tau'$ by $\tau$ in above:
\begin{eqnarray*}
\Delta & \ge &  \textstyle \Ex[\sum_{i=1}^{\tau} v_{t+i} - (\tau+1) (1-\epsilon)\mu + I(\tau=k) (1-\epsilon)\mu | s_t, \val_t]- \Ex[\sum_{i=1}^{\tau} v_{t+i} I(A_{t+i}) | s_t, \val_t]
\end{eqnarray*}
Combining the first and last term in above,  we get $\sum_{i=1}^{\tau} v_{t+i} I(\overline{A_{t+i}})$. Now, \mbox{$v_{t+i} I(\overline{A_{t+i}}) - \mu(1-\rho)$}, $i=1,2,\ldots$ form a martingale, and $\tau$ is a finite 
stopping time ($\tau\le k$), therefore, by Wald's equation,
\begin{center}
$\Ex[\sum_{i=1}^{\tau} v_{t+i}I(\overline{A_{t+i}})| s_t, \val_t] = \Ex[\tau |s_t, \val_t] \Ex[v_{t+1} I(\overline{A}_{t+1}) | s_t, \val_{t}] = \Ex[\tau | s_t, \val_t] \mu(1-\rho)$
\end{center}
In the last expression we used that $A_{t+1}$ and $v_{t+1}$ are independent, given $s_t, \val_t$.
Substituting, we obtain, (in below we drop the conditional on $s_t,a_t$ for notational brevity)
$$
\begin{array}{rcl}
\Delta & \ge &  \textstyle \Ex[\tau]\mu (1-\rho) - (\Ex[\tau]+1) (1-\epsilon)\mu + \Pr(\tau=k) (1-\epsilon)\mu\\
 & = &  \textstyle \Ex[\tau](\epsilon-\rho) \mu - (1-\epsilon)\mu + \Pr(\tau=k) (1-\epsilon)\mu\\
& = &  \textstyle \Ex[\tau] (\epsilon-\rho) \mu - \Pr(\tau<k) (1-\epsilon)\mu
\end{array}
$$
Now, let $X$ be a geometric random variable with success probability $\rho$, then $\tau$  stochastically dominates $\min\{X, k\}$. And, from Claim \ref{claim:geometric} (stated and proved below)
\begin{center}
$\Ex[\tau] \ge \Ex[\min\{X,k\}] =  \frac{1}{\rho} \Pr(X<k) +  \Pr(X\ge k)$,\\
$\Pr(\tau<k) \le \Pr(X<k) = 1-(1-\rho)^{k-1}$
\end{center}
The proof is completed by the following algebraic manipulations:
\begin{eqnarray*}
\Delta & \ge &  \textstyle \Ex[\tau] (\epsilon-\rho) \mu - \Pr(\tau<k) (1-\epsilon)\mu\\
& = &  \textstyle \frac{1}{\rho} \Pr(X<k) (\epsilon-\rho) \mu +  \Pr(X\ge k) (\epsilon-\rho) \mu - \Pr(X<k) (1-\epsilon)\mu\\
& = &  \textstyle \frac{(\epsilon-\rho) \mu}{\rho} - (1-\epsilon)\mu + \Pr(X\ge k) ( -\frac{(\epsilon-\rho) \mu}{\rho}  +(\epsilon-\rho) \mu  + (1-\epsilon)\mu)\\
& = &  \textstyle \frac{\epsilon \mu}{\rho} - (2-\epsilon)\mu + (1-\rho)^{k-1} \mu ( 2 - \rho -\frac{\epsilon }{\rho} )
\end{eqnarray*}
We are given that $\rho \le \frac{\epsilon}{(1-\epsilon)}$. Consider two cases: $2-\rho-\frac{\epsilon}{\rho}> 0$ and $2-\rho-\frac{\epsilon}{\rho}\le 0$. In the first case, the second term above is positive so that $\Delta_{k+1} > \frac{\epsilon \mu}{\rho} - (2-\epsilon)\mu \ge 0$, because $\rho \le \frac{\epsilon}{(2-\epsilon)}$. In the second case, $\Delta_{k+1}$ is minimized for $k=1$, i.e., when $\Delta_{k+1} = \Delta_{2} =  \frac{\epsilon \mu}{\rho} - (2-\epsilon)\mu + \mu ( 2-\rho-\frac{\epsilon }{\rho} ) = (\epsilon-\rho) \mu > 0$.

This proves that $U^t_k(s_t, v_{t}, b'_t) - U^t_k(s_t, v_{t}, b_t) = \Delta >0$ when $\rho \le \frac{\epsilon}{(2-\epsilon)}$, proving a contradiction that $b_t$ is not $k$-lookahead optimal. Thus, the $k$-lookahead optimal bid if exists will ensure that $s_{t+1} \ne \bot$. 

In fact, by induction optimal $k-1$-lookahead bid exists, so that the optimal $k$-lookahead bid for any $t$ such that $s_t \ne \bot$ is given by:
$$ b_t := \arg\max_{b:({\bar b} n+b)/(n+1) \ge (1-\epsilon)\mu} \Ex[\val_t -b+ \max_{b'} U^{t+1}_{k-1}(s_{t+1}, \val_{t+1}, b') | s_t, \val_t],$$
which by applying this lemma for $k-1, k-2, \ldots$ can be derived to be the minimum bid that would keep $s_{t+1}$ as a good state.
\end{proof}

\begin{claim}
\label{claim:border}
Under mechanism $\mthree$ with $\rho \le 
\frac{\epsilon}{2-\epsilon}$, 
an optimal $k$-lookahead bid $b_t$ at time $t$, when starting from any good state $s_t =s \ne \bot$, would achieve at least as much utility as when starting from a borderline state $s_t = s' = ((1-\epsilon)\mu, n)$. That is,
$$ U^t_k(s, \val_t, b_t) \ge U^t_k(s', \val_t, b_t), \forall s\ne \bot, s'=((1-\epsilon)\mu $$
\end{claim}
\begin{proof}
Consider the case when the starting state is a borderline state $s'=((1-\epsilon)\mu, n)$.
Opening up the recursive definition of $k$-lookahead utility, we obtain the following expression in terms of bids $b_{t+1}, \ldots, b_{t+k-1}$  which are optimal $k-1, k-2, \ldots, 1$ lookahead bids respectively. 

\begin{eqnarray*}
 U^t_k(s', \val_t, b_t) & := & \Ex[\sum_{\tau = t}^{t+k} \val_\tau x(s_\tau, b_\tau)- p(s_\tau, b_\tau, x_\tau) | s_t=s', \val_t] 
\end{eqnarray*}
Using Lemma \ref{lem:transitionGood}, the optimal $k$-lookahead bid for any $k\ge 1$ is such that the next state is a good state, so that if the starting state $s'$ is a good state, then so are the states $s_\tau, \tau=t+1,\ldots, t+k$ in the above expression. This further implies that if the starting state is a borderline state $s'=(1-\epsilon)\mu, n$, then the sum of bids $b_t, b_{t+1}, \ldots, b_{t+k-1}$ must be at least $(1-\epsilon)\mu k$.
Since in good state, the allocation is always $1$  and the payment is equal to the bid, we obtain the following upper bound on the utility:
\begin{eqnarray*}
 U^t_k(s', \val_t, b_t) & = & \Ex[\sum_{\tau = t}^{t+k} \val_\tau - b_\tau | s_t, \val_t] \\
& \le & \val_t + \Ex[ \sum_{i=1}^{k} \val_{t+i}] - k(1-\epsilon)\mu.
\end{eqnarray*}
Now, on starting from another good state, say $s = (\avgbid, n') \ne \bot$, since $\avgbid \ge (1-\epsilon)\mu$, the sum of bids  $b_t, b_{t+1}, \ldots, b_{t+k-1}$ needs to be {\it at most} $(1-\epsilon)\mu$ to remain in a good state, and $b_{t+k}=0$ as the optimal myopic  bid (for good state) will be used in this last step. Therefore, for any $s\ne \bot$,
$$ U^t_k(s, \val_t, b_t)  \ge \val_t + \Ex[ \sum_{i=1}^{k} \val_{t+i}] - k(1-\epsilon)\mu \ge  U^t_k(s', \val_t, b_t)$$
\end{proof}
\begin{claim}
\label{claim:geometric}
Let $X$ be a geometric random variable with success probability $\rho$, then 
$$\Ex[\min\{X,k\}] =  \frac{1}{\rho} \Pr(X<k) +  \Pr(X\ge k)$$
\end{claim}
\begin{proof}
\begin{eqnarray*}
\Ex[\min\{X,k\}] & = & \Ex[X I(X<k)] + \Pr(X\ge k) k \\
& = & \sum_{j=1}^{k-1} (1-\rho)^{j-1} \rho j + k (1-\rho)^{k-1}\\
& = & \Ex[X] - \sum_{j=k}^{\infty} (1-\rho)^{j-1} \rho j + k (1-\rho)^{k-1}\\
& = & \Ex[X] -  (1-\rho)^{k-1}  \sum_{j=1}^{\infty} (1-\rho)^{j-1} \rho (j+k-1) + k (1-\rho)^{k-1}\\
& = & \Ex[X] -  (1-\rho)^{k-1}  \Ex[X] - (1-\rho)^{k-1} \sum_{j=1}^{\infty} (1-\rho)^{j-1} \rho (k-1) + k (1-\rho)^{k-1}\\
& = & \frac{1}{\rho} (1-  (1-\rho)^{k-1}) - (1-\rho)^{k-1} (k-1) + k (1-\rho)^{k-1}\\
& = & \frac{1}{\rho} (1-  (1-\rho)^{k-1}) +  (1-\rho)^{k-1} \\
& = & \frac{1}{\rho} \Pr(X<k) +  \Pr(X\ge k) 
\end{eqnarray*}
\end{proof}

\label{app:lower}
\begin{prevproof}{Theorem}{thm:lowerbound}
We will take our distribution $F$ to be the Pareto distribution with parameter $\alpha > 2$, supported on $[1,+\infty)$. In particular, the probability density function is $f(x)={\alpha \over x^{\alpha+1}}$, $x \in [1,+\infty)$. Note that the mean is $\mu={\alpha \over \alpha-1}$ and the variance is $\sigma^2={\alpha \over (\alpha-1)^2 (\alpha-2)}$, which are both finite. It is easy to see that $\revmye=1$, and that $F$ is a decreasing hazard rate distribution, as well as a regular distribution.

Now let $M=(\newstatespace,\trans,\alloc,\price,\newstate_1)$ be a mechanism as defined in Section~\ref{sec:prelim}. Consider a myopic buyer being at state $\newstate_t \in \newstatespace$ of the mechanism at time $t$. Given his realized value $\val_t \sim F$ and facing the (randomized) allocation rule $\alloc(\newstate_t,\cdot)$ and price rule $\price(\newstate_t,\cdot,\cdot)$ of the mechanism in state $\newstate_t$, he would map his value $\val_t$ to some bid $\bid_t$ to maximize his expected utility $\Ex_{\alloc \sim \alloc(\newstate_t,\bid_t)}[\alloc] \cdot \val_t - \Ex_{\alloc \sim \alloc(\newstate_t,\bid_t)}[p(\newstate_t,\bid_t,\alloc)]$.

To prove our result let us suppose that there exists a collection of functions $b_{\newstate}: \mathbb{R} \rightarrow \Delta^{\mathbb{R}^+}$, indexed by states $\newstate \in \newstatespace$, which map a realized value for a myopic bidder to a (potentially randomized) bid, and which are such that the expected revenue of mechanism $\mech$ (average over $T$ rounds) against a myopic buyer using these bidding functions is at least $\epsilon \cdot \revmye - o(1)$, where $o(1)$ is a function that goes to $0$ with $T$. Note that these mappings are only indexed by single states $\newstate$ as they pertain to the behavior of a myopic buyer. Moreover, note that we actually do not need to require that for all $\newstate, \val$, $b_{\newstate}(\val)$ is optimal. The only assumption that we need to make is that, for all $\newstate,\val, \val'$, the distribution over bids $b_{\newstate}(\val)$ does not result in worse utility for a buyer with value $\val$ compared to the distribution $b_{\newstate}(\val')$.  Given this definition, let us also define the effective allocation probability and effective price functions, $\hat{\alloc}: \newstatespace \times \mathbb{R} \rightarrow [0,1]$ and $\hat{\price}: \newstatespace \times \mathbb{R} \rightarrow \mathbb{R}$ respectively, as follows:
$$\forall \newstate,\val: \hat{\alloc}(\newstate, \val) = \Ex_{b \sim b_\newstate(\val), \alloc \sim \alloc(\newstate,b)}[\alloc]~~~~\text{and}~~~~\hat{\price}(\newstate, \val) = \Ex_{{b \sim b_\newstate(\val), \alloc \sim \alloc(\newstate,b)}}[\price(\newstate,b,\alloc)].$$

Via standard argumentation, for all $\newstate \in \newstatespace$, $\hat{\alloc}(\newstate,\cdot)$ and $\hat{\price}(\newstate,\cdot)$ satisfy the incentive compatibility constraint that:
$$\forall \val, \val': \hat{\alloc}(\newstate,\val) \cdot \val - \hat{\price}(\newstate,\val) \ge \hat{\alloc}(\newstate,\val') \cdot \val - \hat{\price}(\newstate,\val').$$ 
Moreover, given that $\mech$ is non-payment forceful, for all $\newstate \in \newstatespace$, we get that
$$\hat{\price}(\newstate,0) = 0.$$

Using Myerson's payment identity, it is standard to argue that any mechanism $(\hat{\alloc}(\newstate,\cdot),\hat{\price}(\newstate,\cdot))$ satisfying the above constraints can be implemented as a distribution over take-it-or-leave-it offers of the item at different prices. That is, there exists a distribution $G_\newstate$ over prices such that the expected revenue and expected buyer utility resulting from $(\hat{\alloc}(\newstate,\cdot),\hat{\price}(\newstate,\cdot))$ can be written as:
\begin{align*}
{\rm Rev}^{\newstate}_{\rm myop} &= \Ex_{\val \sim F}[\hat{\price}(\newstate,\val)] \equiv \Ex_{\val \sim F, \price \sim G_{\newstate}}[\price \cdot 1_{\val \ge \price}] \overset{*}{=} \Ex_{\price \sim G_{\newstate}}\left[{1 \over p^{\alpha-1}}\right];\\
{\rm Ut}^{\newstate}_{\rm myop} &= \Ex_{\val \sim F}[\hat{\alloc}(\newstate,\val) \cdot \val - \hat{\price}(\newstate,\val)] \equiv \Ex_{\val \sim F, \price \sim G_{\newstate}}[(\val-\price) \cdot 1_{\val \ge \price}] \overset{*}{=} \Ex_{\price \sim G_{\newstate}}\left[{1 \over \alpha-1}{1 \over p^{\alpha-1}}\right].
\end{align*}
(In the above, the equalities $\overset{*}{=}$ follow by plugging in for $F$ the distribution defined above.) So, in particular, it follows that ${\rm Ut}^{\newstate}_{\rm myop} = {1 \over \alpha-1} {\rm Rev}^{\newstate}_{\rm myop}$. To summarize, in any state $s \in \newstatespace$, a myopic buyer makes utility that is a factor of $\alpha-1$ smaller than the payment that he makes.

Now recall that our mechanism has expected average per round revenue against a myopic buyer that is at least $\epsilon \cdot \revmye - o(1)$. It follows from the above derivation that it should also then give expected average per round utility at least ${1 \over \alpha-1}\epsilon \cdot \revmye - o(1)$ to a myopic buyer. As an infinite look-ahead buyer (aiming to maximize his utility) can certainly pretend to be myopic, this means that the mechanism must give expected average per round utility at least ${1 \over \alpha-1}\epsilon \cdot \revmye - o(1)$ to an infinite look-ahead buyer. Hence, the expected average per round revenue that the mechanism can get from an infinite look-ahead buyer is at most:
\begin{align*}
\mu - {1 \over \alpha-1}\epsilon \cdot \revmye +o(1) &= \mu - {1 \over \alpha-1}\epsilon \cdot 1+o(1)\text{ }~~~~~~~~~~~~~\text{~~~~~~~~~~~~(recalling that $\revmye=1$)}\\
&= \mu - {\alpha \over \alpha-1} {\epsilon \over \alpha}+o(1)\\
&= \left(1-{\epsilon \over \alpha}\right)\mu +o(1)\text{ }~~~~~~~~~~~~~~~~~~~~~\text{~~~~~~~~~~~~(recalling that $\mu={\alpha \over \alpha-1}$)}
\end{align*}
As $\alpha$ can be made arbitrarily close to $2$, the theorem holds.
\end{prevproof}

\begin{prevproof}{Theorem}{th:expostLB}
We will prove the stated bound using equal revenue distribution. $k$-lookahead buyers optimize $k$-lookahead utility, which in round $t$ is given by 
$$ U_k^t(s_t,v_t,b_t)= (v_t x(s_t,b_t) - p(s_t,b_t)) + \max_{b'} \Ex_{s', v'}[U_{k-1}^{t+1}(s',b', v') | s_t, b_t]$$

Let us first consider the case when the allocation function is deterministic, i.e., $x(s,b) \in \{0,1\}$. Now, for any state $s$, let $B^1_s$ be the set of bids such that $x(s,b)=1$.  Let $V^1_s$ be the set of valuations in the support of $F$ such that given any valuation $v\in V^1_s$, there is at least one utility maximizing bid that gets an allocation in state $s$. More precisely, given a valuation $v$ and state $s$, let $B^*_{s,v} := \{\arg \max_b U(s,v, b)\}$ denote the set of utility maximizing bids. Then, $V^1_s$ is defined the set of valuations $v$ for which $B^1_s \cap B^*_{s,v}$ is non-empty. We observe that the set of utility maximizing allocating bids is the same for all valuations in $V^1_s$. That is, $B^1_s \cap B^*_{s,v} = B^1_{s} \cap B^*_{s,v'} =: B^*_s$, for any $v,v'$.  This is because, for bids $b$ in $B^1_s$, $x(s,b)=1$; therefore for any valuation $v$, the utility maximizing bids in $B^1_s \cap B^*_{s,v}$ are given by set $\{\arg \max -p(s,b)  + \Ex_{s', v'}[\max_{b'} U_{k-1}(s', v',b') | s,b\}$, which does not depend on valuation $v$. 

Now, let $v_{min}=\inf\{v\in V^1_s\}$. Then, for all $b\in B^*_{s}$, by IR property:
$$0\le U^t_k(s,b, v_{min}) = v_{min} - p(s,b) + \Ex[\max_{b'} U^{t+1}_{k-1}(s',b',v') | s,b] \le v_{min} - p(s,b) + k\mu$$
where $k\mu$ is an upper bound on the $(k-1)$-lookahead utility. Thus, rearranging $0 \leq v_{min} - p(s,b) + k\mu$ gives us:

$$p(s,b) \le v_{min}+ k\mu.$$ 

Further from the per-round ex-post IR property, $p(s_t, b_t) \le \val_t$. 
Therefore, expected revenue in round $t$, given  state $s_t=s$ is upper bounded by

$$\Ex[ p(s_t,b_t)|s_t=s] \le \Ex_{v\in V}[\min\{v_{min}+k\mu, v\} I(v\in V^1_s)] \le \Ex_{v\in V}[\min\{v_{min}+k\mu, v\} I(v\ge v_{min})] $$ 
For equal revenue distribution, this is upper bounded by 
$$\int_{v_{min}}^{v_{min}+k\mu} v \frac{1}{v^2}\ dv + (v_{min}+k\mu)\times \frac{1}{v_{min}+k\mu} = \log(v_{min}+k\mu) - \log(v_{\min}) +1 \le \log(k\mu)+1$$ 

{\it Randomized allocation.} The above argument can be extended to $x(s,b)\in [0,1]$, by replacing set $B^1_s$ and $V^1_s$ by $B^x_s$ and $V^x_s$ respectively, defined for every possible value of allocation $x\in [0,1]$. That is, $B^x_s$ is the set of bids such that $x(s,b)=x$.  And, $V^x_s$ is the set of valuations in the support of $F$ such that given any valuation $v\in V^x_s$, there is at least one utility maximizing bid that gets an allocation of $x$ in state $s$. Then, $B^*_{s,v} \cap B^x_s$ denote all the bids which are utility maximizing and get allocation of $x$. Using the same argument as above, this set can be shown to be independent of $v$, i.e., $B^*_{s,v} \cap B^x_s = B^*_{s,x}$. 

Now, for any non-empty $V^x_s$, let $v_{min,x}=\inf\{v\in V^x_s\}$. Then, for all $b\in B^*_{s,x}$, by IR property:
$$0\le U^t_k(s,b, v_{min,x}) = x v_{min, x} - p(s,b) + \Ex[\max_{b'} U^{t+1}_{k-1}(s',b',v') | s,b] \le x v_{min, x} - p(s,b) + k\mu$$
where $U^*_{k-1} \le k\mu$ is an upper bound on the $(k-1)$-lookahead utility. 
so that 
$$p(s,b) \le x v_{min, x}+ k\mu,$$ 

For $s_t, b_t, x_t=x(s_t,b_t)$, if $b_t$ is an optimal bid, then $V^{x_t}_{s_t}$ must be non-empty, and above inequality can be applied, to get $p(s_t,b_t) \le x_t v_{min, x_t}+ k\mu$. Further from the per-round ex-post IR property, $p(s_t, b_t) \le \val_t$. Then, the above argument can be repeated while replacing $v_{min}$ by $v_{min,x_t}$, to obtain the same upper bound.

\begin{eqnarray*}
\Ex[ p(s_t,b_t)|s_t=s] & \le & \Ex_{v\in V}[\min\{x_t v_{min}+k\mu, v\} I(v\in V^{x_t}_{s_t})] \\
& \le & \Ex_{v\in V}[\min\{v_{min,x_t}+k\mu, v\} I(v\ge v_{min,x_t})] 
\end{eqnarray*}
For equal revenue distribution, this is bounded by $\log(v_{min, x_t}+k\mu) - \log(v_{\min,x_t}) +1 \le \log(k\mu)+1$.
\end{prevproof}
\end{document}